%% file: 0-main.tex
\documentclass[11pt]{article}

\PassOptionsToPackage{numbers,sort}{natbib}
\usepackage{lmodern}

\usepackage{fullpage}
\usepackage{amsmath}
\usepackage{bbm}
\usepackage[normalem]{ulem}
\usepackage{verbatim}
\usepackage{amssymb}
\usepackage{authblk}
\usepackage{natbib}

\input{macros}

\input{math}

\title{No-Regret Online Autobidding Algorithms in First-price Auctions}

\author[1]{Yuan Deng\thanks{\texttt{dengyuan@google.com}}}
\author[2]{Yilin Li\thanks{\texttt{ylli25@cse.cuhk.edu.hk}}}
\author[2]{Wei Tang\thanks{\texttt{weitang@cuhk.edu.hk}}}
\author[2]{Hanrui Zhang\thanks{\texttt{hanrui@cse.cuhk.edu.hk}}}

\affil[1]{Google Research}
\affil[2]{Chinese University of Hong Kong}
\date{}
\begin{document}

\maketitle

\begin{abstract}
  Automated bidding to optimize online advertising with various constraints, e.g. ROI constraints and budget constraints, is widely adopted by advertisers. A key challenge lies in designing algorithms for non-truthful mechanisms with ROI constraints. While prior work has addressed truthful auctions or non-truthful auctions with weaker benchmarks, this paper provides a significant improvement: We develop online bidding algorithms for repeated first-price auctions with ROI constraints, benchmarking against the optimal randomized strategy in hindsight. In the full feedback setting, where the maximum competing bid is observed, our algorithm achieves a near-optimal $\tildeO(\sqrt{T})$ regret bound, and in the bandit feedback setting (where the bidder only observes whether the bidder wins each auction), our algorithm attains $\tildeO(T^{3/4})$ regret bound.\footnote{A conference version of this work has been accepted in the proceeding of the 39th Conference on Neural Information Processing Systems (NeurIPS'25).} 
\end{abstract}

\section{Introduction}
\input{1-intro}

\section{Preliminaries}
\label{preliminaries}
\input{2-prelim}

\section{Optimal Randomized Autobidding Strategy}
\label{sec:optimal_random}

\input{3-randomized}



\section{No-Regret Autobidding Algorithm  with Full Feedback}
\label{sec:full_feedback}

\input{4-full-feedback-new}

\section{No-Regret Autobidding Algorithm with Bandit Feedback}
\label{sec:bandit_feedback}

\input{5-bandit-feedback-new}

\newpage
\bibliographystyle{ACM-Reference-Format}
\bibliography{references}

\newpage
\appendix
\section{Missing Proofs in \texorpdfstring{\Cref{sec:optimal_random}}{sec 3}}
\label{apx:proof random}

\input{apx-proof-optimal}

\section{Missing Proofs in \texorpdfstring{\Cref{sec:full_feedback}}{sec 4}}
\label{apx:proof fullfeedback}

\input{apx-proof-fullfeedback}

\section{Missing Proofs in \texorpdfstring{\Cref{sec:bandit_feedback}}{sec 5}}
\label{apx:proof banditfeedback}
\input{apx-proof-banditfeedback}

\end{document}

%% file: macros.tex
\usepackage[utf8]{inputenc} 
\usepackage[T1]{fontenc}    
\usepackage{hyperref}       
\usepackage{url}            
\usepackage{booktabs}       
\usepackage{amsfonts}       
\usepackage{nicefrac}       
\usepackage{xfrac}
\usepackage{microtype}      
\usepackage{xcolor}         

\usepackage{algorithm}
\usepackage{algorithmic}
\newcommand{\Comment}[1]{{\footnotesize\ttfamily\textcolor{blue}{/* #1 */}}}

\usepackage{amsthm}
\usepackage{enumitem}
\usepackage{cleveref}

\usepackage{color-edits}
\addauthor{wt}{purple}
\addauthor{yd}{cyan}

\theoremstyle{plain}
\newtheorem{theorem}{Theorem}[section]
\newtheorem{lemma}[theorem]{Lemma}

\newtheorem{proposition}[theorem]{Proposition}

\theoremstyle{plain}
\newtheorem{definition}{Definition}[section] 
\newtheorem{example}[definition]{Example}

\theoremstyle{plain}

\newtheorem{assumption}{Assumption}

\newcommand{\val}{v}

\newcommand{\valSeq}{V}

\newcommand{\valDist}{H}
\newcommand{\allocPay}{G}
\newcommand{\allocPayConv}{G_{\mathrm{conv}}}
\newcommand{\allocPayConcaveEnvelope}{G_{\mathrm{concave}}}

\newcommand{\optAllocPay}{\bar G}

\newcommand{\xConv}{x_{\mathrm{conv}}}
\newcommand{\yConv}{y_{\mathrm{conv}}}

\newcommand{\CDF}{F}
\newcommand{\CDFConv}{F_{\mathrm{conv}}}
\newcommand{\bidConv}{\bid_{\mathrm{conv}}}

\newcommand{\empiricalCDF}{\widehat{\CDF}}
\newcommand{\optimisticCDF}{\bar{\CDF}}
\newcommand{\estimatedCDF}{\CDF^\dag}

\newcommand{\bid}{b}
\newcommand{\competingbid}{d}
\newcommand{\competingbidDist}{F}
\newcommand{\alloc}{x}
\newcommand{\reward}{r}
\newcommand{\payment}{p}

\newcommand{\TotalRegret}{\mathrm{Regret}}
\newcommand{\TotalReward}{\mathrm{Reward}}
\newcommand{\TotalPayment}{\mathrm{Payment}}
\newcommand{\TotalROIViolation}{\mathrm{\Delta}}
\newcommand{\Algorithm}{\cc{ALG}}
\newcommand{\algWithF}{\cc{ALG1}}

\newcommand{\ConsAlg}{\cc{ALG1}^{\cc{c}}}
\newcommand{\opt}{\cc{OPT}}

%% file: math.tex
\newcommand{\N}{\mathbb{N}}
\newcommand{\eps}{\varepsilon}
\newcommand{\tildeO}{\widetilde{O}}

\newcommand{\prob}[2][]{\text{Pr}\ifthenelse{\not\equal{}{#1}}{_{#1}}{}\!\left[{\def\givenn{\middle|}#2}\right]}
\newcommand{\indicator}[2][]{\mathbf{1}\ifthenelse{\not\equal{}{#1}}{_{#1}}{}\!\left\{{\def\givenn{\middle|}#2}\right\}}
\newcommand{\expect}[2][]{\mathbb{E}\ifthenelse{\not\equal{}{#1}}{_{#1}}{}\!\left[{\def\givenn{\middle|}#2}\right]}

\newcommand{\cc}[1]{\ensuremath{\mathsf{#1}}}

%% file: 1-intro.tex
{\em Autobidding} has become the predominant paradigm in online advertising.  The idea is that advertisers compete for ad opportunities through repeated auctions, where each advertiser employs an automated bidding algorithm (i.e., an autobidder) to bid on behalf of the advertiser.  Autobidders are essentially online constrained optimizers: As auctions happen over time, an autobidder maximizes on the fly a certain objective quantity, e.g., clicks or conversions, subject to constraints, e.g., the total spending cannot exceed a predetermined budget, or the ratio between the return and the investment (i.e., the return-on-investment, or ROI ratio) must be at least a given target ratio.  Normally, advertisers are free to customize these components of autobidders to better serve their own goals.

Over years, advertisers seem to have converged to the following common practice: An overwhelming majority of advertisers use autobidders that behave like {\em ROI-constrained value maximizers}.  This highlights the importance of designing good online bidding algorithms in this very setting.  The problem, roughly speaking, is the following: Initially, an auction mechanism (e.g., first-price or second-price) is chosen externally by the advertising platform, which is used in all subsequent auctions.  At each step, an auction happens, and the algorithm observes the ``value'' of winning in this auction.  Then, the algorithm must immediately submit a bid based on historical observations and the value, after which the algorithm observes the outcome of the auction, i.e., whether it has won and possibly also the highest bid by other advertisers.  In the long run, the algorithm must (approximately) maximize the total value obtained over all auctions, subject to the constraint that the overall ROI ratio is at least a predetermined target ratio.  The key challenge here is that at each step, the algorithm must decide the bid without knowing the future, or even the competing bid in the current auction.  Nonetheless, it must ensure that its decisions are almost optimal from hindsight (i.e., it has {\em no regret}), and the ROI constraint is approximately satisfied over the entire horizon.

Indeed, considerable effort has been invested into this task.  Here, we refrain from a prolonged discussion (see Section~\ref{sec:related_work} for further related work) and focus on the most related results.  \citet{feng2023online} study the problem when the auction mechanism is {\em truthful}, and give an almost optimal no-regret online bidding algorithm.  Despite its undoubted importance, one severe limitation of their result is that it does not apply to some of the most common auction mechanisms in reality, including first-price auctions and generalized second-price auctions --- for these mechanisms are not truthful.  To address this issue, \citet{aggarwal2025no} design alternative no-regret bidding algorithms that work for nontruthful auctions.  However, their algorithms only have no regret against a weaker benchmark, i.e., the optimal {\em Lipschitz-continuous} bidding strategy from hindsight, which in general can be much worse than the unconditionally optimal bidding strategy. 
This leaves the following question open: {\em Is it possible to design a no-regret online bidding algorithm against the unconditionally optimal strategy for nontruthful auction mechanisms?}

\paragraph{Our results.}
We provide an affirmative answer to the above question.  For concreteness, we derive our results for first-price auctions (which is the most representative nontruthful auction mechanism), but our results easily generalize to other auction mechanisms.  We consider two natural feedback models commonly studied in the literature:
\begin{itemize}[topsep=2pt,itemsep=0pt,leftmargin=2.5ex,parsep=1pt]
    \item {\em Full feedback}: the algorithm learns the auction outcome and the competing bid, after each auction;
    \item {\em Bandit (or one-bit) feedback}: the algorithm learns only the outcome of the auction and nothing else after each auction.
\end{itemize}
The latter setting is more challenging than the former, since the algorithm obtains less information.

In the full feedback setting, we present an online bidding algorithm that achieves the {\em optimal} (up to polylog factors) regret bound of $\widetilde{O}(\sqrt{T})$ against the {\em unconditional, possibly randomized, optimal strategy} in hindsight, where $T$ is the time horizon.  This generalizes the result by \citet{feng2023online} and strengthens the one by \citet{aggarwal2025no} in the same setting.  In the bandit feedback setting, we present an algorithm that achieves a regret bound of $\widetilde{O}(T^{3/4})$ against the unconditionally optimal strategy, which strengthens the result by \citet{aggarwal2025no} in the same setting.  These results paint a considerably clearer picture of online bidding algorithms for ROI-constrained value maximizers.

\paragraph{Technique overview.}
The key technical challenge in designing online autobidding algorithms, which has also been observed by \citet{aggarwal2025no}, is that unconditionally optimal bidding strategies appear to have rich structures.
This tends to baffle existing techniques, which normally require the class of strategies optimized over to be ``simple'' and / or ``structured'' in some way.
Below we first briefly review how prior work deals with this challenge.

\citet{feng2023online} restrict their attention to truthful auctions, in which optimal bidding strategies admit a clean characterization: Without loss of generality, they only need to consider ``uniform bid scaling'', i.e., bidding strategies defined by a single multiplier $\theta$, such that the strategy places a bid of $\theta \cdot v$ when the value of winning is $v$.
Restricted to such bidding strategies, one essentially only needs to optimize the multiplier $\theta$, which allows \citet{feng2023online} to essentially reduce their problem (in a nontrivial way) to one that can be handled within the powerful mirror descent framework.
\citet{aggarwal2025no} take an approach that is in a sense more general: They consider bidding strategies that are Lipschitz-continuous, which subsume uniform bid scaling as a subclass.
By doing so, \citet{aggarwal2025no} manage to further incorporate technical insights from the Lipschitz bandits literature, which help establish nontrivial regret bounds for nontruthful auctions.

We take an approach that is in spirit similar to \citep{feng2023online}.
We first make a more general structural observation, that even in nontruthful auctions, optimal (and possibly randomized) bidding strategies are {\em essentially parametrized by a single number}, which roughly corresponds to the marginal ROI of the bidding strategy.
We present an algorithmically compatible characterization of this class of strategies (see Section~\ref{sec:optimal_random}), which allows us to (1) generalize the mirror-descent-based algorithm used in \citep{feng2023online} to nontruthful auctions and (2) without loss of generality, focus on deterministic bidding strategies, provided one condition: The algorithm needs to {\em know the distribution of competing bids}.\footnote{
    Note that the problem remains nontrivial, because the algorithm does not know the distribution of its value.
}
This condition is crucial, since the structure of optimal bidding strategies necessarily depends on the distribution of competing bids, and without knowing the former, one simply cannot guarantee the conditions for the mirror-descent primitive to produce the desired regret bound.\footnote{
    This is not an issue in \citep{feng2023online} since the structure they need is independent of the distribution of competing bids.
}
In our setting, of course the algorithm cannot know this distribution beforehand.
Below we discuss how this requirement can be removed.

We first show that the requirement can be {\em relaxed} without hurting the regret too much.
In particular, we show that if we run the generalized mirror-descent-based algorithm with a slightly inaccurate distribution of competing bids given as input, then its per-round regret can only degrade to the extent that the input distribution is inaccurate.
With that as a primitive, we construct an algorithm that ``bootstraps'' itself without any prior knowledge about the environment in the full feedback setting (see Section~\ref{sec:full_feedback}).
The algorithm ``restarts'' itself periodically as it gathers more information about the distribution of competing bids.
Upon each restart, it runs a fresh instance of the mirror-descent-based primitive with the latest estimated distribution as its input, which allows it to gradually lower the per-round regret over time as the estimation becomes more accurate.
By properly pacing the restarts, we manage to beat the $O(T^{2/3})$ bound via simple exploration-exploitation and achieve an almost optimal bound of $\widetilde{O}(\sqrt{T})$.

In the bandit feedback setting, the above strategy does not work any more, and we have to fall back to simple exploration-exploitation.
In our setting, even this conceptually simple high-level strategy requires nontrivial technical ingredients to implement: In fact, here we need a slightly stronger property of the primitive, because we can no longer guarantee a {\em uniform} estimation error of the distribution of competing bids.
We show that the property holds for a slightly adapted version of the primitive used in the full feedback setting.
With that as a building block, by optimizing the granularity of discretization and the tradeoff between exploration and exploitation, we manage to establish a regret bound of $\widetilde{O}(T^{3/4})$ (see Section~\ref{sec:bandit_feedback}), which matches the guarantee established by \citet{aggarwal2025no} against the weaker Lipschitz-continuous benchmark.
Notably, the ideal $O(\sqrt{T})$ bound is impossible in this setting due to a lower bound of $\Omega(T^{2/3})$ by \citet{aggarwal2025no}.
We leave closing this regret gap as an important furture direction.

\subsection{Related Work}
\label{sec:related_work}

Most closely related to our work is the line of research on online bidding algorithms in repeated auctions, with ROI-constrained value-maximizing bidders.
In the paradigm with ROI constrained value maximizing bidders, \citet{feng2023online} propose an algorithm based on the primal-dual framework proposed by \citet{balseiro2020dual}, which achieves a nearly optimal $\widetilde{O}(\sqrt{T})$ regret bound for truthful auctions.
\citet{castiglioni2024online} propose an algorithm that achieves no regret in repeated, possibly nontruthful, auctions, under the assumption that the spaces of values and bids are both finite.
\citet{aggarwal2025no} design an algorithm which works with nontruthful auctions and continuous value / bid space, and achieves no regret against any Lipschitz-continuous bidding strategy.
Our work generalize / strengthen all these results by achieving the nearly optimal $\widetilde{O}(\sqrt{T})$ regret against the unconditionally optimal bidding strategy without further assumptions.\footnote{
    The $\widetilde{O}(\sqrt{T})$ bound applies to the full feedback setting.
    In the bandit feedback setting, we match the bound by \citet{aggarwal2025no} against a weaker benchmark, thereby also strengthening their result.
}
Also conceptually related is the work by \citet{vijayan2025online}, who consider a ``harder'' setting where the algorithm does not even observe the value of winning, and instead, must learn to estimate this value on the fly.
Their result is orthogonal to ours, in particular because they also assume a finite space of bids.

Beyond autobidding, there is an extremely rich body of research on online bidding algorithms in repeated auctions.
To name a few results: \citet{balseiro2020dual} propose a primal-dual framework and apply their framework to repeated second-price auctions with budget constraints and derive no-regret bidding algorithms.
\citet{wang2023learning} study the same problem in repeated first-price auctions.
\citet{cesa2024role} study online bidding algorithms in first-price auctions without budget or ROI constraints, and pin down the optimal regret in a number of settings.
\citet{kumar2024strategically} present no-regret bidding algorithms that are strategically robust.
\citet{susan2023multi} design no-regret bidding algorithms for multi-platform settings, where the bidding algorithm repeatedly participates in multiple parallel auctions.
More generally, there is a long line of research on no-regret algorithms and dynamics for allocation problems under constraints, e.g., \citep{borgs2007dynamics,balseiro2019learning,castiglioni2022online,fikioris2023approximately,fikioris2023liquid,immorlica2022adversarial,weed2016online,han2020learning,nedelec2022learning,lucier2024autobidders,gaitonde2023budget,balseiro2019contextual,kesselheim2020online}.
We refrain from an extensive discussion since these results concern fairly different problems from ours, both conceptually and technically.

Our results are also conceptually related to the growing literature on the design and analysis of auctions and marketplaces with autobidders.
Prior research in this direction concerns various topics, including the efficiency of traditional auction mechanisms \citep{aggarwal2019autobidding,deng2024efficiency,deng2024gsp,liaw2023efficiency,feng2024strategic}, the design of (approximately) optimal mechanisms \citep{balseiro2021landscape,balseiro2024optimal,lu2023auction,lv2023auction,bei2025optimal,deng2022posted}, etc.
See the recent survey by \citet{aggarwal2024auto} for a comprehensive exposition.

%% file: 2-prelim.tex
\newcommand{\ALG}{\cc{ALG}}
\newcommand{\roi}{\rho}
\newcommand{\jointDist}{\mathcal{D}}
\newcommand{\win}{\cc{win}}

We consider an online bidding setting where a single bidder participates in single-item first-price auctions \footnote{We consider first-price auctions for simplicity. Our results generalize to other non-truthful auctions: The main technical difference would be in our \Cref{thm:optimal_randomized_strategy}, where the construction of $\CDFConv$ should be adapted in accordance with the auction mechanism.}, repeated in $T$ rounds.
At each time step $t\in [T]$, the learner realizes the value $\val_t$ in the current round and then submits a bid $\bid_t$ based on the current value and the historical information up to now. The bidder wins the item if her submitted bid is not less than the (highest) competing bid $\competingbid_t$ ($\competingbid_t \in [0,1]$) in this round, 
and we denote by $\alloc_t(\bid_t,\competingbid_t) = \indicator{\bid_t \ge \competingbid_t}$ the ex post allocation outcome of the current round.
If the bidder wins the item, she pays a price $\bid_t$ which equals to her submitted bid, otherwise she pays $0$. Therefore, the payment function, denoted by $\payment_t(\bid_t, \competingbid_t)$, is $\payment_t(\bid_t, \competingbid_t)=\bid_t \cdot \indicator{\bid_t \ge \competingbid_t}$.
We focus on value-maximizing bidder where the bidder's utility at time $t$ is given by $\reward_t(\bid_t,\competingbid_t)=\val_t\cdot \alloc_t(\bid_t,\competingbid_t)$. 

\paragraph{Benchmark and regret definitions.}
We assume that the sequence of value $\val_t$ is drawn independently and identically (i.i.d.) from an unknown distribution, whose CDF is denoted as $\valDist$. 
The bidder's objective is to 
design an online bidding algorithm $\ALG$ that maximizes her
total realized utility subject to a Return-On-Investment (ROI) constraint.\footnote{
    For pedagogical reasons, we focus exclusively on ROI constraints.
    Our results and analysis can be readily generalized to the setting with additional budget constraints, e.g., through techniques introduced in \citep{balseiro2024field}.
}
Namely, her total utility must be at least $\roi$ times her total payment: $\roi \cdot \sum\nolimits_{t\in[T]} \payment_t(\bid_t, \competingbid_t) \leq \sum\nolimits_{t\in[T]} \val_t\cdot \alloc_t(\bid_t, \competingbid_t)$ where $\roi > 0$ is the target ROI ratio of the bidder.
\begin{equation}
\begin{aligned}
    \label{eq:opt}
    \max \quad &  \expect[]{\sum\nolimits_{t\in[T]} \val_t\cdot \alloc_t(\bid_t,\competingbid_t)}
    \\
    \textrm{s.t.} \quad & \expect[]{\roi \cdot \sum\nolimits_{t\in[T]} \payment_t(\bid_t, \competingbid_t)} \leq \expect[]{\sum\nolimits_{t\in[T]} \val_t\cdot \alloc_t(\bid_t, \competingbid_t)}~,
\end{aligned} 
\tag{$\mathcal{P}$}
\end{equation}
where the expectation is over the possible randomness of the bid sequence $\{\bid_t\}$ and the randomness of the maximum competing bid $\{\competingbid_t\}$. 
W.l.o.g., we assume that $\roi = 1$.\footnote{For $\roi \neq 1$, we can rescale bidder's value to be $\val_t\cdot \roi$, and all our analysis/results can be carried over similarly.}

We assume that the (maximum) competing bid $\competingbid_t$ at each round is i.i.d.\ realized from an unknown competing bid distribution $\competingbidDist$.
In this context, the allocation and the bidder's payment in round $t$ are determined by her submitted bid $\bid_t$: 
the allocation function 
$\alloc_t(\bid_t, \competingbid_t)=\CDF(\bid_t)$ and the payment function 
$\payment_t(\bid_t, \competingbid_t) = \bid_t \cdot \CDF(\bid_t)$.

We denote the bidder's value sequence by $\valSeq^T = (\val_t)_{t\in[T]}$.
An online bidding algorithm takes input of the bidder's historical information and the realized value at the current round, and outputs a (possibly randomized) bid.
Given an online algorithm $\ALG$, let $\TotalReward(\ALG \mid (\competingbidDist, \valSeq^T))$ (resp.\ $\TotalPayment(\ALG \mid (\competingbidDist, \valSeq^T))$) denote the bidder's total expected utility (resp.\ total expected payment) given the value sequence $\valSeq^T$ and competing bid distribution $\competingbidDist$:
\begin{align*}
    \TotalReward(\ALG \mid (\competingbidDist, \valSeq^T))
    & = \expect[(\bid_t) \sim \ALG]{\sum\nolimits_{t\in[T]} \val_t\cdot \CDF(\bid_t)}~;\\
    \TotalPayment(\ALG \mid (\competingbidDist, \valSeq^T))
    & = \expect[(\bid_t) \sim \ALG]{\sum\nolimits_{t\in[T]} \bid_t\cdot \CDF(\bid_t)}~.
\end{align*}
Thus, the ROI violation is given by 
$\TotalROIViolation(\ALG \mid (\competingbidDist, \valSeq^T))
= \TotalPayment(\ALG \mid (\competingbidDist, \valSeq^T)) - \TotalReward(\ALG \mid (\competingbidDist, \valSeq^T))$.

Let 
$\opt$ be the corresponding optimal bidding algorithm in hindsight for the program \ref{eq:opt}, i.e., $\opt$ has knowledge about the competing bid distribution $F$ and the realized value sequence $\valSeq^T$. 
Thus, an algorithm's regret relative to $\opt$ under the input value sequence $\valSeq^T$ is
\[
    \TotalRegret(\ALG \mid (\competingbidDist, \valSeq^T)) = \TotalReward(\opt \mid (\competingbidDist, \valSeq^T)) - \TotalReward(\ALG \mid (\competingbidDist, \valSeq^T))~.
\]
We measure the performance of an algorithm $\ALG$ by its expected regret over the distributions $\CDF$ and $\valDist$, defined as follows:
\begin{equation} \label{equation:total_regret}
\begin{aligned}
    \TotalRegret(\ALG \mid (\competingbidDist, \valDist))
    = \expect[\valSeq^T \sim \valDist^T]{\TotalRegret(\ALG \mid (\competingbidDist, \valSeq^T))}~.
\end{aligned}
\end{equation}
and expected ROI violation over the distribution $\CDF$ and $\valDist$: 
\begin{align*}
    \TotalROIViolation(\ALG \mid (\competingbidDist, \valDist))
    = \expect[\valSeq^T \sim \valDist^T]{\TotalROIViolation(\ALG \mid (\competingbidDist, \valSeq^T))}~.
\end{align*}

\paragraph{The feedback models.} 
We now specify the feedback that the bidder receives at the end of each round. In this paper, we consider the following two feedback models.
\begin{itemize}[topsep=2pt,itemsep=0pt,leftmargin=2.5ex,parsep=1pt]
    \item \textbf{Full feedback:} At the end of round $t$, the bidder observes the competing bid $d_t$.
    \item \textbf{Bandit feedback:} At the end of round $t$, the bidder only receives a binary signal $\win_t\in\{0, 1\}$ indicating whether she won the item or not, i.e., 
    $\win_t = \indicator{\competingbid_t \le \bid_t}$.
\end{itemize}

%% file: 3-randomized.tex
In this section, we first present the bidder's optimal bidding strategy in hindsight, where the sequence of values $\valSeq^T$ is known. 
We then design an online bidding algorithm that has low regret and low ROI violation when competing bid distribution $\competingbidDist$ is known, while value distribution $\valDist$ remains unknown.

We first define the following allocation-payment curve, which will be helpful for our analysis.
\begin{definition}[Allocation-payment curve $\allocPay$]
\label{defn:alloc-pay curve}
Given a competing bid distribution $\competingbidDist\in\Delta([0, 1])$,
an allocation-payment curve $\allocPay:[0, 1] \rightarrow [0, 1]$ is defined as: $G(\bid \cdot \CDF(\bid)) = \CDF(\bid)$ for all $b \in [0, 1]$.
\end{definition}

The allocation-payment curve $\allocPay$ captures the relationship between the payment and the allocation (as well as the reward, which is proportional to the allocation when the value of winning is fixed) induced by any bid $b$.

\subsection{Characterizing Optimal Bidding Strategy}
\input{3-1-randomized}

\subsection{No-Regret Autobidding Algorithm for Known Competing Bid Distribution}
\label{subsec:algo known F}
\input{3-2-algo-new}

%% file: 3-1-randomized.tex
\label{sec:optimal_randomized_strategy}

We start with the following example which shows the  suboptimality of deterministic bidding strategies.
\begin{example}[Suboptimality of deterministic strategy]
\label{ex:subopt deterministic}
Consider following competing bid distribution: $\prob{\competingbid=0} =\prob{\competingbid=1} = \sfrac{1}{2}$.
The bidder always has a value $\val=\sfrac{1}{2}$.
Let $T = 1$, and the target ROI ratio $\roi = 1$.
One optimal deterministic bidding strategy would be to always submit a bid $\bid = 0$, which yields an expected reward $\sfrac{1}{4}$ and an expected payment $0$ --- note that restricted to deterministic bidding, there is no way to utilize this ``slackness'' in the ROI constraint.\footnote{
    Here and in the rest of the paper, we assume ties are broken in favor of the bidding algorithm.
    This makes no difference when the distribution of competing bids is continuous.
}
On the other hand, consider the following randomized bidding strategy:
$\prob{\bid=0} = \sfrac{2}{3}, 
\prob{\bid =1} = \sfrac{1}{3}$.
This strategy yields an expected reward $\sfrac{1}{3}$ and an expected payment $\sfrac{1}{3}$.
Both strategies satisfy the ROI constraint, while the randomized strategy gives a strictly higher reward.
\end{example}

\Cref{ex:subopt deterministic} shows that to maximize the bidder's expected reward, it may be necessary to solve program \ref{eq:opt} over all randomized bidding strategies. 
This is (at least superficially) incompatible with the primal-dual framework that underlies all previous results on online autobidding, which presents new technical challenges.
To address this issue, we present a constructive reduction that brings us back to the deterministic world, where we can cast the primal-dual framework more naturally.
In particular, we show that given any input value sequence $\valSeq^T$ and any competing bid distribution $\competingbidDist$, it is always possible to construct another distribution $\CDFConv$ that captures essentially the same tradeoff between reward and payment as $\CDF$ induces, with one key difference: With $\CDF$ replaced by $\CDFConv$, we can focus our attention to deterministic bidding strategies without loss of generality.


\begin{theorem}[Optimal randomized bidding strategy]
\label{thm:optimal_randomized_strategy} 
Fixing any input value sequence $\valSeq^T$. 
For any competing bid distribution $\competingbidDist$, there exists a distribution $\CDFConv\in\Delta([0, 1])$ such that: For any (randomized) strategy under $V^T$ and $F$, there is a deterministic strategy under $V^T$ and $\CDFConv$ inducing reward and payment that are both no worse, and vice versa.
\end{theorem}

In the proof of \Cref{thm:optimal_randomized_strategy}, we explicitly construct such distribution $\CDFConv$: 
the distribution $\CDFConv$ is constructed such that its allocation-payment curve, denoted by $\allocPayConv$, is {\em essentially} the concave envelope of the allocation-payment curve $\allocPay$ induced by the original competing bid distribution $\competingbidDist$.
With \Cref{thm:optimal_randomized_strategy}, we can simplify our algorithm design as follows: 
instead of designing an online bidding algorithm that has low regret against the optimal randomized bidding strategy under $(\competingbidDist, \valSeq^T)$, 
we can turn to designing an algorithm that has low regret against the optimal deterministic strategy under the corresponding $(\CDFConv, \valSeq^T)$.

We conclude this subsection with a proof sketch of \Cref{thm:optimal_randomized_strategy}.
All the missing proofs for the technique lemmas are deferred to \Cref{apx:proof random}.
\begin{proof}[Proof of \Cref{thm:optimal_randomized_strategy}]
Fixing any input value sequence $\valSeq^T$.
For any competing bid distribution $\competingbidDist$,
the proof of \Cref{thm:optimal_randomized_strategy} consists of three main steps:
\begin{itemize}
[topsep=2pt,itemsep=0pt,leftmargin=2.5ex,parsep=1pt]
    \item 
    \textbf{Step 1 -- Prove that the ``concave envelope''
    $\allocPayConv$ \footnote{
        Technically, $\allocPayConv$ might differ from the actual concave envelope of $G$, but this can happen only on the part of $\allocPayConv$ that never matters for our purposes.
        See the proof of \Cref{lem:new_F_with_randomized_strategy} for details.
    }
    of allocation-payment curve $\allocPay$ attains the maximum expected reward (see \Cref{lem:concave_envelope_G})}:
    In our first step, we show that given any fixed value $\val$, any (possibly randomized) bidding strategy that leads to an expected payment of $x$ has an expected reward at most $\val\cdot \allocPayConv(x)$. 
    
    \item 
    \textbf{Step 2 -- Construct the distribution $\CDFConv$ (see \Cref{lem:new_F_with_randomized_strategy})}:
    In this step, we construct such distribution $\CDFConv$ that corresponds to an allocation-payment curve $\allocPayConv$, which is essentially the concave envelope of the allocation-payment curve $\allocPay$ induced by $\competingbidDist$.
    \item 
    \textbf{Step 3 -- Prove the reward equivalence (see \Cref{lem:random_combination})}: In this step, we show that any reasonable randomized bidding strategy under $(\competingbidDist, \valSeq^T)$ is equivalent to generating a bid on the allocation-payment curve $\allocPayConv$, which corresponds to a deterministic bidding strategy under $(\CDFConv, \valSeq^T)$, and vice versa. \hfill \qedhere
\end{itemize}
\end{proof}

%% file: 3-2-algo-new.tex
In this subsection, we describe a no-regret online bidding algorithm (see \Cref{alg:randomized_strategy_under_F}) when the competing bid distribution $\competingbidDist$ is known to the bidder, while the value distribution $\valDist$ still remains unknown. 
This algorithm will serve as a building block for our algorithm design when the both $\competingbidDist$ and $\valDist$ are unknown to the bidder.

We first note that after introducing the Lagrangian multiplier $\lambda \ge 0$,  we can reformulate the program \ref{eq:opt} as follows
\begin{equation}
\label{eq:optimization_with_lambda}
\begin{aligned} 
  \max ~ \expect{\sum\nolimits_{t\in[T]}  \val_t\cdot \competingbidDist(\bid_t) + \min\nolimits_{\lambda \ge 0} \lambda \cdot \sum\nolimits_{t\in[T]} (\val_t\cdot \competingbidDist(\bid_t) - \bid_t \cdot \competingbidDist(\bid_t))}~.
\end{aligned}   
\end{equation}
Following a similar dual-prime framework proposed by \cite{balseiro2020dual}, we proceed by choosing a bid $\bid_t$ at round $t$ to maximize the objective in Eqn.~\eqref{eq:optimization_with_lambda} using a fixed Lagrangian multiplier $\lambda_t$ and then update $\lambda_t$ to $\lambda_{t+1}$ use the observed feedback.
Notice that given a fixed $\lambda_t$ at round $t$, the problem in Eqn.~\eqref{eq:optimization_with_lambda} is equivalent to
\begin{equation} 
    \label{equation:bid_t}
    \begin{aligned}
      \max ~ \expect{\dfrac{1+\lambda_t}{\lambda_t}\cdot \val_t\cdot \competingbidDist(\bid)- \bid \cdot \competingbidDist(\bid)}~.
    \end{aligned} 
\end{equation}
As we show in \Cref{ex:subopt deterministic}, the optimal bidding strategy may be randomized. 
Thus, using \Cref{thm:optimal_randomized_strategy}, instead of directly solving problem Eqn.~\eqref{equation:bid_t}, we consider solving a bid $\widetilde \bid_t$ that maximizes $\frac{1+\lambda_t}{\lambda_t}\cdot \val_t\cdot \CDFConv(\bid)- \bid \cdot \CDFConv(\bid)$.
We then use \Cref{lem:random_combination} to construct the equivalent randomized bidding strategy to be implemented in practice. 
We then update the dual variable $\lambda_t$ using the a generalized negative entropy in the mirror descent framework to adjust the penalty for the future rounds.
We summarize the algorithm details in \Cref{alg:randomized_strategy_under_F}.

\begin{algorithm}
\caption{No-regret bidding algorithm for known $\competingbidDist$}
\label{alg:randomized_strategy_under_F}
\begin{algorithmic}[1]
\REQUIRE Time horizon $T$, competing bid distribution $\CDF$, the value sequence $V^T$ (revealed sequentially).
\ENSURE Initial dual variable $\lambda_1=1$ and the dual mirror descent step size $\alpha=\sfrac{1}{\sqrt{T}}$

\STATE Let $\CDFConv$ be the distribution defined as in 
\Cref{thm:optimal_randomized_strategy} (its construction is in \Cref{lem:new_F_with_randomized_strategy}).
\FOR{$t \gets 1$ to $T$}
    \STATE Observe the value $\val_t$.
    \STATE 
    \label{algo1:compute bid}
    Compute $\widetilde \bid_t \gets \textrm{argmax}_{\bid\in[0, 1]} \left[\frac{1+\lambda_t}{\lambda_t}\cdot \val_t\cdot \CDFConv(\bid)- \bid \cdot \CDFConv(\bid)\right]$
    \STATE 
    Given $\widetilde \bid_t$, construct the equivalent randomized bidding strategy mentioned in \Cref{thm:optimal_randomized_strategy} and realizes a bid $\bid_{t}$.
    \Comment{Such randomized bidding strategy can be constructed efficiently, see \Cref{lem:random_combination}.}
    \STATE Submit the bid $\bid_t$ (to the outer environment or algorithm).
    \STATE Update $g_t(\widetilde b_t)=\val_t\cdot \CDFConv(\widetilde b_t) - \widetilde b_t \cdot \CDFConv(\widetilde b_t)$. 
    \STATE Update $\lambda_{t+1}=\lambda_t \cdot \textrm{exp}[-\alpha \cdot g_t(\widetilde b_t)]$. \label{equation_of_lambda}
\ENDFOR
\RETURN The bid sequence $\{\bid_t\}_{t\in[T]}$.
\end{algorithmic}
\end{algorithm}

The performance of \Cref{alg:randomized_strategy_under_F} is detailed below, whose analysis follows similarly to \cite{balseiro2020dual,feng2023online}.
We defer the proof to \Cref{apx:proof random}.
\begin{proposition}
\label{prop:performance algo known F}
For any unknown value distribution $\valDist$, the expected regret of \Cref{alg:randomized_strategy_under_F} under any competing bid distribution $\competingbidDist$ is bounded by $\TotalRegret(\text{\Cref{alg:randomized_strategy_under_F}} \mid (\CDF, \valDist)) = O(\sqrt{T})$,
and the expected ROI violation is bounded by $\TotalROIViolation(\text{\Cref{alg:randomized_strategy_under_F}} \mid (\CDF, \valDist))  
\le 2\log T \sqrt{T}$.
\end{proposition}

%% file: 4-full-feedback-new.tex
In \Cref{sec:optimal_random}, we present an algorithm
(see \Cref{alg:randomized_strategy_under_F}) when the competing bid distribution is known that achieves low expected regret and ROI violation.
In practice, both the competing bid distribution $\competingbidDist$ and the value distribution $\valDist$ may be unknown, and the bidder needs to use the observed feedback to learn these uncertainties by exploring different bids.
In this section, we consider a full feedback model where the bidder can observe the realized competing bid $\competingbid_t$ at the end of each round, and we provide an algorithm (see \Cref{alg:randomized_strategy}) with low expected regret and low ROI violation.

\paragraph{Algorithm descriptions.}
At a high-level, \Cref{alg:randomized_strategy} operates as a stage-based algorithm.  
In each stage $m\in[M]$ where $M$ is the number of stages determined shortly, we use the competing bids observed in all previous stages to maintain an empirical CDF $\empiricalCDF_m$ for the underlying competing bid distribution $\competingbidDist$.
Let $T_m\in \N^+$ be the number of rounds in stage $m$ and  $N_m=\sum\nolimits_{m'\in [m-1]} T_{m'}$ be the number of total rounds right before the stage $m$.
We maintain $\empiricalCDF_m$ as follows:
\begin{align}
    \label{eq:empiricalCDF}
    \empiricalCDF_m(\bid) = \frac{1}{N_m}\sum\nolimits_{i\in[N_m]}
    \indicator{\bid \ge \competingbid_{i}},
    \quad \bid\in[0, 1]~,
\end{align}
To account for the estimation error and maintain optimism in our bidding strategy, we further define the optimistic CDF $\optimisticCDF_m$ as
\begin{align}
    \label{eq:optimisticCDF}
    \optimisticCDF_m(\bid) 
    = \left(\empiricalCDF_m(\bid) + \eps_m \right ) \wedge 1~,\quad \bid\in[0, 1]~.
\end{align}
where $\eps_m$ is a parameter carefully chosen to balance the exploration and exploitation tradeoff.\footnote{
    As a shorthand, we write $x \wedge y := \min\{x, y\}$.}
We then implement \Cref{alg:randomized_strategy_under_F} by feeding the distribution $\optimisticCDF$ (while the actual competing bids are still generated according to the true unknown competing bid distribution $\competingbidDist$). 
We summarize the algorithm details in \Cref{alg:randomized_strategy}.

\begin{algorithm}
\caption{No-regret bidding algorithm with full feedback}
\label{alg:randomized_strategy}
\begin{algorithmic}[1]
\REQUIRE 
Time horizon $T$, the value sequence $V^T$ (revealed sequentially).
\REQUIRE Let $M=\lceil \log(T+1) \rceil$ be the number of stages.
\ENSURE Let $\optimisticCDF_1$ be an arbitrary distribution. $N_1 \gets 0$ 

\FOR{$m \gets 1$ to $M$}
    \STATE Current stage length: $T_m \gets 2^{m-1}$.
    
    \STATE Feed the \Cref{alg:randomized_strategy_under_F} with inputs: $T \gets T_m, \competingbidDist \gets \optimisticCDF_{m}$.
    \Comment{Note that we compute the submitted bid $\bid_t$ and update dual variable $\lambda_t$ from \Cref{alg:randomized_strategy_under_F} based on $\optimisticCDF_m$, while the actual competing bids are generated according to $\competingbidDist$.}
    \STATE $N_{m+1} \gets T_m + N_m$ 
    \STATE Optimism parameter $\eps_{m+1} \gets \frac{\log (N_{m+1})}{\sqrt{N_{m+1}}}$.
    \STATE 
    Update the empirical CDF $\empiricalCDF_{m}$ to be $\empiricalCDF_{m+1}$ according to Eqn.~\eqref{eq:empiricalCDF}.
    \STATE Obtain $\optimisticCDF_{m+1}$ defined as in Eqn.~\eqref{eq:optimisticCDF}
    
\ENDFOR
\RETURN The bid sequence $\{\bid_t\}_{t\in[T]}$.
\end{algorithmic}
\end{algorithm}

\paragraph{Algorithm performance.}
We summarize the performance of \Cref{alg:randomized_strategy} as follows.
\begin{theorem}
\label{thm:randomized_algorithm_full_feedback} 
For any unknown competing bid distribution $\competingbidDist$ and unknown value distribution $\valDist$, under the setting with full feedback, the expected regret of
\Cref{alg:randomized_strategy}  is bounded by $\TotalRegret(\mathrm{\Cref{alg:randomized_strategy}} \mid (\CDF, \valDist)) = \widetilde O(\sqrt{T})$,
and the expected ROI violation is bounded by $\TotalROIViolation(\mathrm{\Cref{alg:randomized_strategy}} \mid (\CDF, \valDist))  = \widetilde O(\sqrt{T})$.
\end{theorem}

We conclude this section by providing a proof sketch for \Cref{thm:randomized_algorithm_full_feedback}.
All missing proofs of the technical lemmas are deferred to \Cref{apx:proof fullfeedback}. 
\begin{proof}[Proof of \Cref{thm:randomized_algorithm_full_feedback}] We first discuss the regret and ROI violation incurred with any fixed value sequence $\valSeq^T$. 
The analysis of regret and ROI violation of \cref{alg:randomized_strategy} consists of three main steps:
\begin{itemize}
[topsep=2pt,itemsep=1pt,leftmargin=2.5ex,parsep=1pt]
    \item 
    \textbf{Step 1 -- Accuracy of the optimistic CDF $\optimisticCDF$ (see \Cref{lem:full_feedback-distance_between_optimistic_CDF_and_CDF})}: In this step, we proved that when $\eps_m = \Theta\left(\frac{\log N_m}{\sqrt{N_m}}\right)$, the CDF $\optimisticCDF_m$ (defined in Eqn.~\eqref{eq:optimisticCDF}) serves as an optimistic estimate to the true underlying competing bid distribution ($\optimisticCDF_m \ge \CDF$ with high probability), while still remains close, specifically we have $\sup_b |\optimisticCDF_m(b)-\CDF(b)|\le 2\eps_m$ with high probability. 
    \item
    \textbf{Step 2 -- 
    The performance approximation when implementing \Cref{alg:randomized_strategy_under_F} with estimated distributions (see \Cref{lem:estimated_distribution_aganst_optimal_strategy})}: 
    In this step, we establish the performance guarantee (evaluating in an environment with competing bid distribution $\competingbidDist$) of the bids generated by feeding \Cref{alg:randomized_strategy_under_F} with a distribution $\estimatedCDF$ that satisfies $\estimatedCDF \ge \competingbidDist$ and $\sup_b |\estimatedCDF_m(b)-\CDF(b)|\le 2\eps$,  
    \begin{align*}
        \TotalRegret(\algWithF(\estimatedCDF) \mid (\CDF, \valSeq^T)) 
        = O(\sqrt{T}) + 2\eps T, 
        ~ \TotalROIViolation(\algWithF(\estimatedCDF)\mid (\CDF, \valSeq^T)) 
        = \widetilde O (\sqrt{T})+ 2\eps T~.
    \end{align*}
    Here we denote by $\algWithF(\estimatedCDF)$ the bids generated by feeding \Cref{alg:randomized_strategy_under_F} with a distribution $\estimatedCDF$.
    In stage $m$, when we use the estimated optimistic $\optimisticCDF_m$ as input for \Cref{alg:randomized_strategy_under_F}, \Cref{lem:estimated_distribution_aganst_optimal_strategy} demonstrates that both the regret and ROI violation will be amplified by $T\cdot 2\eps_m$ compared to using the true CDF $\CDF$ as input.
    \item
    \textbf{Step 3 -- Aggregating the regret and ROI violation over $M$ stages (see \Cref{lem:full_feedback-Regret} and \Cref{lem:full_feedback-ROI})}: With $\eps_m$ selected in Step 1, \Cref{lem:full_feedback-Regret} proved that the regret of \Cref{alg:randomized_strategy} is bounded by $\widetilde O(\sqrt T)$. \Cref{lem:full_feedback-ROI} proved that the ROI violation is bounded by $\widetilde O(\sqrt T)$.
\end{itemize}
Thus, when a value sequence $\valSeq^T$ is i.i.d.\ drawn from $\valDist^T$,  the expected regret and ROI violation over the distribution $\CDF$ and $\valDist$ satisfies
\begin{align*}
    \TotalRegret(\mathrm{\Cref{alg:randomized_strategy}} \mid (\competingbidDist, \valDist))
    & = \expect[\valSeq^T \sim \valDist^T]{\TotalRegret(\mathrm{\Cref{alg:randomized_strategy}} \mid (\competingbidDist, \valSeq^T))} 
    = \widetilde O(\sqrt T)~.\\
    \TotalROIViolation(\mathrm{\Cref{alg:randomized_strategy}} \mid (\competingbidDist, \valDist))
    & = \expect[\valSeq^T \sim \valDist^T]{\TotalROIViolation(\mathrm{\Cref{alg:randomized_strategy}} \mid (\competingbidDist, \valSeq^T))}
    = \widetilde O(\sqrt T)~. \qedhere
\end{align*}
\end{proof}

%% file: 5-bandit-feedback-new.tex
\newcommand{\conservBid}{\bid^{\cc{c}}}
\newcommand{\conservCDF}{\optimisticCDF^{\cc{c}}}

In this section, we consider a more challenging learning setting in which the bidder receives only bandit feedback at the end of each round -- specifically, whether they won the item or not.
This bandit feedback model is common in practice and has been extensively studied in the literature (see, e.g., \citep{cesa2024role,aggarwal2025no}).
The main results in this section are an online autobidding algorithm (see \Cref{alg:randomized_strategy_with_bandit_feedback}) with the provable performance (see \Cref{thm:bandit_feedback-regret_ros}) under such bandit feedback.

\paragraph{Algorithm descriptions.}
At a high-level, our proposed \Cref{alg:randomized_strategy_with_bandit_feedback} operates as an explore-then-exploit type algorithm. 
The algorithm starts with an exploration stage that explores different bids to obtain sufficient knowledge about the underlying competing bid distribution. 
Then, the algorithm implements a ``conservative'' variant of \Cref{alg:randomized_strategy_under_F} with the estimated competing bid distribution over the remaining rounds.

In the exploration phase, to explore different bids, we discretize the bid range $[0,1]$ into $K$ equal intervals with grid points $\{c_k\}_{k\in[K]\cup\{0\}}$, where $c_k=\sfrac{k}{K}$ for $k\in[K]\cup\{0\}$. 
We then bid each grid point $c_k$ consecutively for $M$ rounds, and then use the observed bidding outcomes to compute an empirical estimate $\empiricalCDF$ for the underlying competing bid distribution $\CDF$: for $\ell \in [1:K-1]$
\begin{align}
    \empiricalCDF(c_\ell) 
    = 
    \max\nolimits_{k\in [\ell]\cup\{0\}} \frac{1}{M} \cdot \sum\nolimits_{t\in [k\cdot M+1: (k+1) \cdot M]} 
    \indicator{\bid_t \ge \competingbid_t}~,
    \label{eq:empirical CDF one-bit}
\end{align}
and we let $\empiricalCDF(c_K) = 1$.
We also extend the definition of $\empiricalCDF$ to any bid $\bid \in [0,1)$ as a step function:
$\empiricalCDF(\bid) = \empiricalCDF(c_k)$ if $\bid\in[c_k, c_{k+1})$.
After exploring all the grid bids, we can have sufficient observations to maintain a good estimate $\empiricalCDF$ for the underlying the competing bid distribution. 
Similar to \Cref{alg:randomized_strategy}, we also define the optimistic CDF $\optimisticCDF$ as follows:
\begin{align}
    \label{eq:optimisticCDF one-bit}
    \optimisticCDF(\bid) =
    \left(\empiricalCDF(c_k) + \eps\right) \wedge 1, \quad \text{if } \bid\in[c_k, c_{k+1}) ~,
\end{align}
where $\eps$ again is a parameter carefully chosen to balance the exploration and exploitation tradeoff.

In \Cref{alg:randomized_strategy} with full feedback, after obtaining the optimistic CDF $\optimisticCDF$, we directly 
implement \Cref{alg:randomized_strategy_under_F} assuming the underlying competing bid distribution is $\optimisticCDF$. 
However, this approach may fail under our bandit feedback setting, as we can no longer guarantee a uniform error bound on $\optimisticCDF$ (i.e., $\sup_b |\CDF(b) - \optimisticCDF(b)|$ might be large).
Instead, in the exploitation phase of \Cref{alg:randomized_strategy_with_bandit_feedback}, we implement \Cref{alg:randomized_strategy_under_F} for the remaining rounds in a ``conservative'' way.
In particular, we define a conservative version of the optimistic CDF $\conservCDF$ as follows:
\begin{align}
    \label{eq:conservCDF}
    \conservCDF(\bid) = \optimisticCDF\left((\bid+\sfrac{1}{K}) \wedge 1\right), \quad \bid \in[0, 1]~.
\end{align}
We then feed \Cref{alg:randomized_strategy_under_F} with the input $T - K\cdot M$ as the time horizon, and $\conservCDF$ as if the underlying competing bid distribution is $\conservCDF$ to generate a sequence of bids $\{\conservBid_t\}$ for the remaining rounds.
Instead of directly using these bids, we submit the bids also in a conservative way. 
In particular, we bid $\{\bid_t\}$ where each bid $\bid_t \gets (\conservBid_t + \sfrac{1}{K}) \wedge 1$.

\begin{algorithm}
\caption{No-regret bidding algorithm with bandit feedback}
\label{alg:randomized_strategy_with_bandit_feedback}
\begin{algorithmic}[1]
\REQUIRE Time horizon $T$, the value sequence $V^T$ (revealed sequentially).
\ENSURE Initialize dual variable $\lambda_1=1$,
$K, M$.

\Comment{Exploration phase}
\STATE Discretize $[0,1]$ to obtain grid points $\{c_0, c_1, c_2, \ldots, c_K\}$.
\FOR{$k \gets 0$ to $K-1$}
    \STATE Bid $c_k$ consecutively for $M$ time rounds. \STATE Observe the allocation outcomes 
    $\indicator{\bid_t \ge \competingbid_t}_{t\in [k\cdot M+1: (k+1) \cdot M]}$.
\ENDFOR
\STATE 
Obtain the empirical CDF $\empiricalCDF$, the optimistic CDF $\optimisticCDF$ and the conservative CDF $\conservCDF$ defined as in Eqn.~\eqref{eq:empirical CDF one-bit}, \eqref{eq:optimisticCDF one-bit} and ~\eqref{eq:conservCDF}, respectively.

\Comment{Exploitation phase}
\STATE
Feed the \Cref{alg:randomized_strategy_under_F} with inputs: $T\gets T - K\cdot M, \competingbidDist \gets \conservCDF$. 
\STATE Let $\{\conservBid_t\}$ be the bids generated from   \Cref{alg:randomized_strategy_under_F} with these input specifications.
\STATE
Submit the bids $\{\bid_t\}$ over remaining rounds where each bid $\bid_t \gets (\conservBid_t + \sfrac{1}{K}) \wedge 1$. 

\RETURN The bid sequence $\{\bid_t\}_{t\in[T]}$
\end{algorithmic}
\end{algorithm}

\paragraph{Algorithm performance.}
The performance of \Cref{alg:randomized_strategy_with_bandit_feedback} is summarized as follows:
\begin{theorem}
\label{thm:bandit_feedback-regret_ros} 
For any unknown competing bid distribution $\competingbidDist$ and unknown value distribution $\valDist$, under the bandit feedback, the expected regret of
\Cref{alg:randomized_strategy_with_bandit_feedback} with $M = \Theta(\sqrt{T}), K = \Theta(T^{\sfrac{1}{4}})$ is bounded by $\TotalRegret(\mathrm{\Cref{alg:randomized_strategy_with_bandit_feedback}} \mid (\CDF, \valDist)) = \widetilde O(T^{\sfrac{3}{4}})$,
and the expected ROI violation is bounded by $\TotalROIViolation(\mathrm{\Cref{alg:randomized_strategy_with_bandit_feedback}} \mid (\CDF, \valDist))  = \widetilde O(T^{\sfrac{3}{4}})$.

\end{theorem}
The proof of \Cref{thm:bandit_feedback-regret_ros} and all relevant lemmas (with proofs)  are deferred to \Cref{apx:proof banditfeedback}.

%% file: apx-proof-optimal.tex
\subsection{Analysis of Randomized Autobidding Strategy}
\begin{lemma}
\label{lem:monoton allocPay}
For any $\competingbidDist$, the allocation-payment curve $\allocPay$ is non-decreasing.
\end{lemma}

\begin{proof}[Proof of \Cref{lem:monoton allocPay}]
Since $\bid \geq 0$, we have that $\bid \cdot \CDF(\bid)$ remains non-decreasing in $\bid$. 
Concretely, consider $\allocPay(x_1) = \CDF(\bid_1) \le \CDF(\bid_2)=\allocPay(x_2)$. Due to the monotonicity of $\CDF(\bid)$, we have $\bid_1 < \bid_2$. Therefore, $x_1 = b_1 \cdot \CDF(\bid_1) < \bid_2 \cdot \CDF(\bid_2)=x_2$.
Hence, the curve $\allocPay$ is a non-decreasing function. \qedhere
\end{proof}

\begin{lemma}
\label{lem:concave_envelope_G}
Given any competing bid distribution $\competingbidDist$, 
let $\allocPay$ be its allocation-payment curve and let $\allocPayConcaveEnvelope(x)$ be its concave envelope. 
Fix any given value $\val\in[0, 1]$.
For any (possibly randomized) bidding strategy that leads to an expected payment of $x$, then its expected reward is at most $\val\cdot \allocPayConcaveEnvelope(x)$. 
\end{lemma}
\begin{proof}[Proof of \Cref{lem:concave_envelope_G}]


Let us fix an arbitrary competing bid distribution $\competingbidDist$.
Fix a value $\val$, and consider any randomized combination of bids $\{\bid_i\}_{i=1}^M$ with the corresponding payments $\{x_i\}_{i=1}^M$ where $x_i = \bid_i\cdot \competingbidDist(\bid_i)$ and the probability $\{p_i\}_{i=1}^M$, where $\sum {p_i}= 1$. 
The expected payment from this randomized bidding strategy is:
    \[\payment =  \sum\nolimits_{i\in[M]} p_i \bid_i \CDF(\bid_i)= \sum\nolimits_{i\in[M]} p_i x_i~.\]
The expected reward $r$ from this randomized bidding strategy is: 
    \[\reward = \val \cdot \sum\nolimits_{i\in[M]} p_i \CDF(\bid_i) = \val \cdot \sum\nolimits_{i\in[M]} p_i \allocPay(x_i)~.\]
By the definition of concave envelope, this point $x=\sum\nolimits_{i\in[M]} p_i \cdot x_i, y=\sum\nolimits_{i\in[M]} p_i \allocPay(x_i)$  will be inside the convex hull. 
Therefore, the expected reward from this randomized bidding strategy is at most $\val \cdot \allocPayConcaveEnvelope(x)$. \qedhere
\end{proof}

\begin{lemma}
\label{lem:new_F_with_randomized_strategy}
Given any $\competingbidDist$, defining $b_0 = \inf \{b \mid F(b) > 0\}$, there exists a distribution $\CDFConv\in\Delta([0, 1])$ such that its allocation-payment curve $\allocPayConv$ has the following properties:
\begin{itemize}
[topsep=2pt,itemsep=1pt,leftmargin=2.5ex,parsep=1pt]
    \item $\allocPayConv(x) = 0$, for $x \in [0, b_0 \cdot F(b_0))~;$
    \item $\allocPayConv$ is the concave envelope of the allocation-payment curve $\allocPay$ induced by $\competingbidDist$, for $x \in [b_0 \cdot F(b_0), 1]~.$

\end{itemize}
Here, $x$ represents the input variable of the allocation-payment functions.
\end{lemma}

\begin{proof}[Proof of \Cref{lem:new_F_with_randomized_strategy}]
We prove this lemma statement by explicitly constructing such distribution $\CDFConv$.

Note that the concave envelope of $\allocPay$, denoted by $\allocPayConcaveEnvelope$, inherits the monotonic property of $\allocPay$. Suppose that there is a point $(x',y')$ on $\allocPayConcaveEnvelope$ where the slope of $\allocPayConcaveEnvelope$ becomes 0. Then for all $x \geq x'$, the slope remains 0 due to the convexity of $\allocPayConcaveEnvelope$. Consequently, $\allocPayConcaveEnvelope$ is constant for $x \geq x'$. Since we know $\allocPayConcaveEnvelope(1)=1$, it follows that $\allocPayConcaveEnvelope(x)=1$ for all $x \geq x'$.

Hence, considering the definition of $\allocPayConv$ and the property of the concave envelope,  we can split it into three parts: 
\begin{itemize}[topsep=2pt,itemsep=1pt,leftmargin=2.5ex,parsep=1pt]
\item For $x \in [0, \bid_0 \cdot \CDF(\bid_0))$, $\allocPayConv(x)=0~$;
\item For $x \in [\bid_0 \cdot \CDF(\bid_0), x')$, $\allocPayConv = 
\allocPayConcaveEnvelope$ is strictly increasing and thus invertible. We define the invertible function of $\allocPayConv(\cdot)$ as 
\begin{align*}
    \beta(y) \triangleq \allocPayConv^{-1}(y),~ 
    y_0 \triangleq \allocPayConv(\bid_0 \cdot \CDF(\bid_0)) = \CDF(\bid_0)~;
\end{align*} 

\item For $x \in [x', 1]$, $\allocPayConv(x)=1~$.
\end{itemize}



The concrete construction process is split into 3 parts based on the value of $x$:
\paragraph{C1: When $x \in [0,  \bid_0 \cdot \CDF(\bid_0))$.}
By the definition of $\allocPayConv(x)$, we can easily derive that $\CDFConv(\bid) = 0$ and $\bid \in [0, \bid_0)$;
\paragraph{C2: When $x \in [\bid_0 \cdot \CDF(\bid_0), x')$.}
Assume that $\CDFConv$ exists, for any $\bid_0 \leq \bid_1 < \bid_2 < x'$, if $\CDF(\bid_1) = \CDF(\bid_2)$, then $\allocPayConv (\bid_1 \cdot \CDFConv (\bid_1)) = \allocPayConv (\bid_2 \cdot \CDFConv (\bid_2)) $ while $\bid_1 \cdot \CDFConv (\bid_1) \neq \bid_2 \CDFConv (\bid_2)$, which violates the strict monotonicity of $\allocPayConv$. So we can infer that for any $\bid_0 \leq b_1 < b_2 < x'$, $\CDF(\bid_1) \ne \CDFConv(\bid_2)$. This proves that if $\CDFConv$ exists, the invertible function of $\CDFConv$ exists. 

When $x \in [0, x')$, $\allocPay(x) \in [y_0, 1]$. Since $\beta(y) =\bid \cdot y$, we have 
    \[\bid = \frac{\beta(y)}{y}  = \CDFConv^{-1}(y)~,\]
therefore,
    \[ \CDFConv^{-1}(y) = \frac{\beta(y)}{y}~.\]
Define $\gamma(y) \triangleq \frac{\beta(y)}{y}$.
Since $\beta(y)$ is known, we can get $\CDFConv^{-1}(y)=\gamma(y)$.
Note that $\gamma(y)$ must be strictly monotonic in $y$ in order for $\CDFConv$ to exist. Indeed, if $\gamma$ fails to be consistently monotonic, then we cannot assign a unique $\mathrm{d}$ to each $y$, and thus cannot construct a well-defined $\CDFConv$.
When $\gamma$ is strictly monotonic, we can define its inverse $\gamma^{-1}$. Consequently, on any interval where $\beta(y)/y$ is strictly increasing, we obtain
\begin{align*}
    \CDFConv(\bid) = \gamma^{-1}(\bid)~.
\end{align*}
Our next step is to prove $\frac{\beta(y)}{y}$ is strictly increasing.
Set $y_0 < a_1 < a_2 \leq 1$, then $k \in (0,1)$ exists such that $a_1 = k\cdot a_2 + (1-k)\cdot y_0$. Due to the convexity of $\beta(y)$ and $\beta(y_0)=0$
\begin{align*}
  \frac{\beta(a_1)}{a_1}  =\frac{\beta(k\cdot a_2 + (1-k)\cdot y_0)}{k\cdot a_2 + (1-k)\cdot y_0} 
  &  \frac{k \cdot \beta(a_2) + (1-k) \cdot \beta(y_0)}{k\cdot a_2 + (1-k)\cdot y_0}
  \leq \frac{k \cdot \beta(a_2)}{k\cdot a_2 + (1-k)\cdot y_0}~.
\end{align*} 
Since $y_0 = \CDF(\bid_0) > 0$ from the definition of $\bid_0$,
\begin{align*}
  \frac{\beta(a_1)}{a_1}  < \frac{k \cdot \beta(\cdot a_2)}{k\cdot a_2} = \frac{\beta(a_2)}{a_2}~.
\end{align*} 
We find that $\frac{\beta(y)}{y}$ is strictly increasing, therefore, $\CDFConv(\bid)$ exists, for $x \in [\bid_0 \cdot \CDF(\bid_0), x')$

\paragraph{C3: When $x \in [x', 1]$}
Since $\allocPayConv=1$, we have
    \[\allocPayConv(x) = \CDFConv(\bid) =1~,\]
when $\bid=\bid\CDFConv(\bid)=x$.

Combining C1 and C2, we obtain a piecewise definition of $\CDFConv(\bid)$:
\begin{align*}
\CDFConv(\bid)
= \begin{cases}
0~, & \bid \in [0, \bid_0)~,\\
\gamma^{-1}(\bid)~, & \bid \in [\bid_0, x')~, \\
1~, & \bid \in [x', 1]~.
\end{cases}
\end{align*}
Here, $\bid_0$ is defined in \Cref{lem:new_F_with_randomized_strategy}; $x'$ is the variable of concave envelope $\allocPayConcaveEnvelope$ such that $(x', \allocPayConcaveEnvelope(x'))$ is the point where the slope of $\allocPayConcaveEnvelope$ becomes 0. 

The proof then finishes. \qedhere
\end{proof}

\begin{lemma}
\label{lem:random_combination}
Fix any value $\val\in[0, 1]$ and any competing bid distribution $\competingbidDist$.
Let $\CDFConv$ be defined as in \Cref{lem:new_F_with_randomized_strategy}.
Given any bid $\bid \in [0, 1]$, there exists a distribution over two bids $\bid_1, \bid_2$ with probability $p_1, p_2$ (respectively), such that
the bidder's expected reward (resp.\ expected payment) when submitting the deterministic bid $\bid$ under $\CDFConv$ matches the expected reward (resp.\ expected payment) when submitting the randomized bids $\bid_1$ w.p.\ $p_1$ and $\bid_2$ w.p.\ $p_2$ under $\competingbidDist$,
namely, we have $\val \CDFConv(\bid) = p_1 \cdot \val \competingbidDist(\bid_1) + p_2 \cdot \val \competingbidDist(\bid_2)$ (resp.\ $\bid \CDFConv(\bid) = p_1 \cdot \bid_1 \competingbidDist(\bid_1) + p_2 \cdot \bid_2 \competingbidDist(\bid_2)$).
\end{lemma}
\begin{proof}[Proof of \Cref{lem:random_combination}]
According to the definition of $\CDFConv$ in \Cref{lem:new_F_with_randomized_strategy}, when $\bid \in [0, \bid_0)$ ($\bid_0 = \inf \{\bid \mid \CDF(\bid) > 0\}$), both $\CDFConv(b)$ and $\CDF(\bid)$ equal zero, which cause zero payment and reward. Thus, either the (possibly randomized) optimal strategy in $\CDF$ or the deterministic strategy in $\CDFConv$ will not select them as bids. For the proof of this lemma, we can consider only the bid where $\bid \in [\bid_0, 1]$, that is, the phase in the concave envelope $\allocPayConcaveEnvelope$.

After we calculate the optimal $\bidConv$ with $\CDFConv$ and obtain the corresponding point $(\xConv, \yConv)$ on the concave envelope $\allocPayConcaveEnvelope$, we can express this point as a convex combination of points on the original curve $\allocPay$. Specifically, we can find two bids $x_1, x_2 \in [\bid_0 \cdot \CDF(\bid_0), 1]$  and the non-negative coefficient $p_1$ and $p_2$ satisfying $p_1+p_2=1$, such that
\[
(\xConv, \yConv)
=
\bigl(p_1 x_1 + p_2 x_2,
p_1 \allocPay(x_1) + p_2 \allocPay(x_2)\bigr)~.
\]
Since we know $y_1 =  \allocPay(x_1) = \allocPayConcaveEnvelope(x_1)$ and $y_2 = \allocPay(x_2) = \allocPayConcaveEnvelope(x_2)$, combining with $\xConv=p_1 x_1 + p_2 x_2$ and $p_1 + p_2 =1$, $x_1, x_2$ can be determined. 


After obtaining $(x_1, \allocPay(x_1))$ and $(x_2, \allocPay(x_2))$, we can further determine $\bid_1=\frac{x_1}{\allocPay(x_1)}$ with $\CDF(\bid_1) = \allocPay(x_1)$, and $\bid_2=\frac{x_2}{\allocPay(x_2)}$ with $\CDF(\bid_2) = \allocPay(x_2)$. 
Thus, bids $\bid_1$ and $\bid_2$ with probability $p_1$ and $p_2$ define the optimal randomized strategy on the original competing bid distribution $\CDF$, and it achieves the same expected reward and expected payment with the deterministically bidding $\bidConv$ under the $\CDFConv$: 
\[v \CDFConv(\bid) = v \cdot (p_1 y_1 + p_2 y_2) = p_1 \cdot \val \competingbidDist(\bid_1) + p_2 \cdot \val \competingbidDist(\bid_2)~,  \]
and
\[
    \bid \CDFConv(\bid) = p_1 x_1 + p_2 x_2 = p_1 \cdot \bid_1 \competingbidDist(\bid_1) + p_2 \cdot \bid_2 \competingbidDist(\bid_2)~. \qedhere
\] 
\end{proof}
\subsection{Proofs of \Cref{prop:performance algo known F} }
\begin{proof}[Proof of ROI violation in \Cref{prop:performance algo known F}]
We discuss the expected ROI violation by first fixing a value sequence $\valSeq^T$ and then calculating the expected value when considering any value sequence $\valSeq^T$ i.i.d. drawn from $\valDist^T$.

For a fixed value sequence $\valSeq^T$, we have 
$\bid_t = \textrm{argmax}_{\bid\in[0, 1]} \left[\frac{1+\lambda_t}{\lambda_t}\cdot \val_t\cdot \CDFConv(\bid)- \bid \cdot \CDFConv(\bid)\right]$ at each round $t$.
Since $\bid \in [0,1]$, when $\bid =0$, the function $\frac{1+\lambda_t}{\lambda_t}\cdot \val_t\cdot \CDFConv(\bid)- \bid \cdot \CDFConv(\bid) \ge 0$. Therefore, $\max~\left[\frac{1+\lambda_t}{\lambda_t}\cdot \val_t\cdot \CDFConv(\bid)- \bid \cdot \CDFConv(\bid)\right] \ge 0$.
Thus, 
    \[
        \frac{1+\lambda_t}{\lambda_t}\cdot \val_t\cdot \CDFConv(\bid_t)- \bid_t \cdot \CDFConv(\bid_t) \ge 0~.
    \]
We further derive
    \[b_t \leq \frac{1+\lambda_t}{\lambda_t}\cdot \val_t~.\]
To simplify notation, we denote by $g_t(\bid_t)$ the difference between the bidder's utility and payment:
\begin{equation} \label{equation:g_t(b_t)}
\begin{aligned} 
  g_t(\bid_t)\triangleq 
  \val_t\cdot \CDFConv(\bid_t) - \bid_t \cdot \CDFConv(\bid_t)~.
\end{aligned} 
\end{equation}
Therefore, $g_t(b_t)$ satisfies:
  \[
      g_t(\bid_t)=\val_t\cdot \CDFConv(\bid_t) - \bid_t \cdot \CDFConv(\bid_t) = (\val_t- \bid_t) \cdot \CDFConv(b_t) \ge -\frac{\val_t}{\lambda_t}\cdot \CDFConv(b_t) \ge -\frac{1}{\lambda_t}~.
    \]
Thus, 
  \[
     \max\left(-1, -\frac{1}{\lambda_t}\right) \leq g_t(b_t) \leq \val_t \cdot \CDFConv(b_t)~.
  \]
Given $\alpha = \frac{1}{\sqrt{T}}$ and the total ROI violation $ \TotalROIViolation(\mathrm{\Cref{alg:randomized_strategy_under_F}} \mid (\CDF, \valSeq^T)) = -\sum\nolimits_{t\in[T]}g_t(b_t)$:
\begin{itemize}[topsep=2pt,itemsep=1pt,leftmargin=2.5ex,parsep=1pt]
    \item If $-\sum_{t\in[T]} g_t(b_t) \le \log T \sqrt{T}$, the total ROI violation is is bounded by:
    \[
       \TotalROIViolation(\mathrm{\Cref{alg:randomized_strategy_under_F}} \mid (\CDF, \valSeq^T)) = - \sum\nolimits_{t\in[T]}g_t(b_t) \le \log T \sqrt{T}~.
    \]
    \item If  $-\sum_{t\in[T]} g_t(b_t) > \log T \sqrt{T}$, we can find $T'$ which is the last time that $\sum_{t\in[T]} g_t(b_t) \le \log T \sqrt{T}$, so we know for any $t>T'+1$, the dual variable $\lambda_t$ must be larger than $T$ since
        \[ \lambda_{t}=\textrm{exp}\left[-\alpha \sum_{t'=1}^{t-1}g_t'(b_t')\right]  > \textrm{exp}[\alpha\sqrt{T}] = T~. \]
    Thus,
        \[ -g_t(b_t) \leq \frac{1}{\lambda_t} \le \frac{1}{T}~.\]
    The ROI violation is bounded by 
    \begin{align*}
        \TotalROIViolation(\mathrm{\Cref{alg:randomized_strategy_under_F}} \mid (\CDF, \valSeq^T)) 
          &   = - \sum\nolimits_{t\in[T]}g_t(b_t)\\
          &  =  - \sum\nolimits_{t\in[T']}g_t(b_t) +  - \sum\nolimits_{t > T'}g_t(b_t) \\
          &  \le \log T \sqrt{T} + 2 \\
          &  \le 2\log T \sqrt{T}~.
    \end{align*}
\end{itemize}
Thus, for $\valSeq^T$ i.i.d.\ drawn from $\valDist^T$, the expected ROI violation of the algorithm under the distribution $\valDist$ satisfies
\begin{align*}
    \TotalROIViolation(\mathrm{\Cref{alg:randomized_strategy_under_F}} \mid (\CDF, \valDist)) 
    =  \expect[\valSeq^T \sim \valDist^T]{\TotalROIViolation(\mathrm{\Cref{alg:randomized_strategy_under_F}} \mid (\CDF, \valSeq^T))} 
    \le 2\log T \sqrt{T}~. 
\end{align*}
\end{proof}
\begin{proof}[Proof of expected regret in \Cref{prop:performance algo known F}]
We split the proof into two main steps:
\paragraph{Step 1: To prove the regret is bounded by $ \expect[\valSeq^T \sim \valDist^T]{\sum_{t\in[T]} \lambda_t \cdot g_t(b_t)}$.}

The reward function $\reward_t(\bid)=\val_t\cdot \CDF(b)$. For $\lambda \ge 0$, define a function of $\lambda$ as:
\[
  f_t^*(\lambda) \triangleq \max_{\bid \in [0,1]} [\reward_t(b) + \lambda \cdot g_t(\bid)]~,
\]
where $g_t(\bid)$ is defined in Eqn. \ref{equation:g_t(b_t)}.
Then, we define
\[
    \bar{D}(\lambda \mid (\CDF, \valDist)) \triangleq \expect[\valSeq^T \sim \valDist^T]{f_t^*(\lambda)}~.
\]
We first get the upper bound of the optimal randomized strategy $\opt$. Due to the weak duality, we know that the primal question will be less or equal to the dual problem. Since $\val_t$ is i.i.d generated from the distribution $\valDist$, the optimal strategy of each round is independent. We can prove that the total reward of $T$ rounds for the optimal strategy will be bounded by:
\begin{align*} 
  \TotalReward(\opt \mid (\CDF, \valDist)) \le T \cdot \min_{\lambda \ge0} \bar{D}(\lambda \mid (\CDF, \valDist))~,
\end{align*} 
where $\opt$ is the corresponding optimal bidding algorithm in hindsight. 
Second, we prove the lower bound of \Cref{alg:randomized_strategy_under_F}. when we get the dual variable $\lambda_t$ at round $t$, the lower bound of \Cref{alg:randomized_strategy_under_F}'s reward is 
\begin{align*} 
  \TotalReward(\mathrm{\Cref{alg:randomized_strategy_under_F}} \mid (\CDF, \valDist)) = T \cdot \bar{D}(\bar \lambda_T\mid (\CDF, \valDist)) - \expect[\valSeq^T \sim \valDist^T]{\sum\nolimits_{t\in[T]} \lambda_t \cdot g_t(b_t)}~, 
\end{align*} 
where $\bar\lambda_T=\frac{1}{T}\sum_{t\in[T]} \lambda_t$.

To prove this, at round $t$, from the optimization algorithm, we can get that
\[
    \bid_t = \textrm{argmax} [\val_t\cdot \CDF(\bid) + \lambda_t \cdot g_t(\bid)]~.
\] 
Therefore, the reward in round $t$ satisfies
\begin{align*} 
\reward_t(b_t)   & = \val_t\cdot \CDF(\bid_t) \\
  & = \max{[\val_t\cdot \CDF(\bid) + \lambda_t \cdot g_t(\bid)]} - \lambda_t \cdot g_t(\bid_t) \\
  & = f_t^*(\lambda_t) - \lambda_t \cdot g_t(b_t)~.
\end{align*} 
Taking expectation on both sides given the output until iteration $t-1$, denote the condition as $\sigma_{t-1}:$
\[
    \expect{\reward_t(b_t) \mid \sigma_{t-1}} =  \expect{f_{t}^*(\lambda_t) \mid \sigma_{t-1}} - \expect{\lambda_t \cdot g_t(b_t) \mid \sigma_{t-1}}~.
\]
Because $\sigma_{t-1}$ decides $\lambda_t$, $\expect{f_t^*(\lambda_t) \mid \sigma_{t-1}}=\expect[\valSeq^{t-1} \sim \valDist^{t-1}]{f_t^*(\lambda_t)}$. We have
\begin{align*} 
  \expect{\reward_t(b_t) \mid \sigma_{t-1}} =  \bar{D}(\lambda_t \mid (\CDF, \valDist)) - \expect{\lambda_t \cdot g_t(b_t) \mid \sigma_{t-1}}~.
\end{align*} 
The sum of the rewards over $T$ rounds,
\begin{align*}
    \TotalReward(\mathrm{\Cref{alg:randomized_strategy_under_F}} \mid (\CDF, \valDist)) 
      & = \expect[\valSeq^T \sim \valDist^T]{ \sum\nolimits_{t\in[T]}\reward_t(\bid_t)} \\
      & =  \sum\nolimits_{t\in[T]}\bar{D}(\lambda_t \mid (\CDF, \valDist)) - \expect[\valSeq^T \sim \valDist^T]{ \sum\nolimits_{t\in[T]}\lambda_t \cdot g_t(b_t)}~.
\end{align*}
By the Definition of regret in Eqn.~\eqref{equation:total_regret}, we have:
\begin{align*}
    \TotalRegret(\mathrm{\Cref{alg:randomized_strategy_under_F}}\mid (\competingbidDist, \valDist))
      & = \expect[\valSeq^T \sim \valDist^T]{\TotalReward(\opt \mid (\competingbidDist, \valSeq^T)) - \TotalReward(\mathrm{\Cref{alg:randomized_strategy_under_F}}\mid (\competingbidDist, \valSeq^T)}\\
      & = \expect[\valSeq^T \sim \valDist^T]{\TotalReward(\opt \mid (\CDF, \valSeq^T)} - \expect[\valSeq^T \sim \valDist^T]{\TotalReward(\mathrm{\Cref{alg:randomized_strategy_under_F}} \mid (\CDF, \valSeq^T)} \\
      & \le T \cdot \min_{\lambda \ge0} \bar{D}(\lambda \mid (\CDF, \valDist)) -  \sum\nolimits_{t\in[T]}\bar{D}(\lambda_t \mid (\CDF, \valDist)) + \expect[\valSeq^T \sim \valDist^T]{ \sum\nolimits_{t\in[T]}\lambda_t \cdot g_t(b_t)} \\
      & \le \expect[\valSeq^T \sim \valDist^T]{ \sum\nolimits_{t\in[T]}\lambda_t \cdot g_t(b_t)}~.
\end{align*}
\paragraph{Step 2: To prove, for any sequence of input value $\valSeq^T$, the sequences of $\{\bid_t\}_{t=1}^{T}$ and $\{\lambda_t\}_{t=1}^{T}$ generated by \Cref{alg:randomized_strategy_under_F} satisfies}
\begin{equation} \label{equation:labmda_g_t}
    \begin{aligned}
        &  \sum\nolimits_{t\in[T]}\lambda_t \cdot g_t(\bid_t) = O(\sqrt{T})~.
    \end{aligned} 
\end{equation}
Since $\lambda_{t+1}=\lambda_t \cdot \textrm{exp}[-\alpha \cdot g_t(b_t)$], $\alpha g_t$ can be represented as:
\[\alpha g_t(\bid_t)= \log(\lambda_t/\lambda_{t+1})~.\]
In the mirror descent framework, we have $V_h(y,x)=h(y)-h(x) - h'(x) \cdot(y-x)$ (Bregman divergence) and $h(u) = u\log u - u$ (generalized negative entropy); thus, $V_h(y,x)$ can be rewritten as:
\[
    V_h(y,x)=y\log(y/x)-y+x ~.
\]
Therefore, we can rewrite $\alpha g_t(b_t)\cdot\lambda_t$ as: 
\begin{align*}
  \alpha g_t(b_t)\cdot\lambda_t   &  =\alpha g_t(b_t)\cdot(\lambda_t - \lambda_{t+1}) + \alpha g_t(b_t) \cdot \lambda_{t+1} \\
    & =\alpha g_t(b_t)\cdot(\lambda_t-\lambda_{t+1}) + \log(\lambda_t/\lambda_{t+1}) \cdot \lambda_{t+1} \\
    & = (\alpha g_t(b_t) + 1)\cdot(\lambda_t-\lambda_{t+1}) - V_h(\lambda_{t+1}, \lambda_t)~. 
\end{align*} 
Due to the local strong convexity of $V_h$, for any $x, y > 0$,
\begin{align*}
    & V_h(y,x)=y\log(y/x)-y+x \geq \frac{1}{2\max(x,y)}\cdot(y-x)^2~,
\end{align*} 
so $\alpha g_t(b_t)\cdot\lambda_t$ satisfies: 
\[
    \alpha g_t(b_t)\cdot\lambda_t \le (\alpha g_t(b_t) + 1)\cdot(\lambda_t-\lambda_{t+1}) - \frac{(\lambda_{t+1}-\lambda_t)^2}{2\max(\lambda_t,\lambda_{t+1})}~.
\]
Based on Cauchy-Schwarz inequality,
\[
    \left(\frac{\lambda_t - \lambda_{t+1}}{\sqrt{2}\max(\lambda_t, \lambda_{t+1})}\right)^2 + \left( \frac{\alpha g_t(b_t)}{\sqrt{2}} \right)^2 \ge \alpha g_t(b_t) \cdot \frac{\lambda_t - \lambda_{t+1}}{\max(\lambda_t, \lambda_{t+1})}~. 
\]
Therefore,
\[
    \frac{(\lambda_t - \lambda_{t+1})^2}{2\max(\lambda_t, \lambda_{t+1})} + \frac{1}{2}\alpha^2 g_t^2(b_t)^2 \cdot \max(\lambda_t, \lambda_{t+1}) \ge \alpha g_t(b_t) \cdot (\lambda_t - \lambda_{t+1})~.
\]
$\alpha g_t(b_t)\cdot\lambda_t$ satisfies: 
\begin{equation} \label{equation:alpha_g_lambda}
  \begin{aligned}
  \alpha g_t(b_t)\cdot\lambda_t \le (\lambda_t - \lambda_{t+1})+\frac{1}{2}\alpha^2g_t(b_t)^2\cdot \max(\lambda_t, \lambda_{t+1})~. 
\end{aligned}  
\end{equation}
We know that $g_t(b_t)$ is bounded by $\max(-1, -\sfrac{1}{\lambda_t}) \leq g_t(b_t) \leq \val_t \cdot x_t(b_t)$. We now bound the term $\frac{1}{2}\alpha^2g_t(b_t)^2\cdot \max(\lambda_t, \lambda_{t+1})$ in a case-wise manner.

\textbf{C1: When we assume $g_t(\bid_t) \ge 0 $.} $g_t(\bid_t) \leq 1 $ with $\alpha = \frac{1}{\sqrt{T}}$ imply $\alpha g_t(b_t) \leq 1$, so $\textrm{exp}(-\alpha g_t(b_t)) \leq 1- \frac{\alpha g_t(\bid_t)}{2} $
\begin{align*}
    & \lambda_{t+1}=\lambda_{t} \textrm{exp}[-\alpha g_t(b_t)] \leq \lambda_{t} \cdot (1 - \alpha g_t(b_t)/2) = \lambda_t - \frac{1}{2}\alpha \lambda_t\cdot g_t(b_t)~.
\end{align*} 
We can infer that $\lambda_t \geq \lambda_{t+1}$. Combine with the bound of $g_t(b_t)$: $0 \leq g_t(b_t) \leq 1 $ to obtain:
\begin{align*}
    & \frac{1}{2}\alpha^2g_t(b_t)^2 \cdot \max(\lambda_t, \lambda_{t+1}) = \frac{1}{2}\alpha^2g_t(b_t)^2 \cdot \lambda_t \leq \alpha \frac{\alpha g_t(b_t)\lambda}{2} \leq \alpha(\lambda_t - \lambda_{t+1})~.
\end{align*} 
\textbf{C2: When we assume $g_t(\bid_t) < 0 $.} $g_t(b_t) \geq \max(-1, -\frac{1}{\lambda_t})$ and $\lambda_t \leq \lambda_{t+1}$
\begin{align*}
    & \frac{1}{2}\alpha^2g_t(b_t)^2 \cdot \max(\lambda_t, \lambda_{t+1}) = \frac{1}{2}\alpha^2g_t(b_t)^2 \cdot \lambda_{t+1} \leq \alpha^2 \cdot \frac{1}{\lambda_t}\cdot \lambda_{t+1}~.
\end{align*} 
Since $g_t(b_t) \geq -1 $ with $\alpha = \frac{1}{\sqrt{T}}$, it implies $\alpha g_t(b_t) \geq -1$. Therefore,
\begin{align*}
    & \lambda_{t+1}=\lambda_{t} \textrm{exp}[-\alpha g_t(b_t)] \leq  e\lambda_t~.
\end{align*} 
We have
\begin{align*}
  \frac{1}{2}\alpha^2g_t(\bid_t)^2 \cdot \max(\lambda_t, \lambda_{t+1}) \le 2\alpha^2~.
\end{align*} 
When $\alpha g_t(\bid_t) \in [-1, 0]$, $\textrm{exp}[-\alpha g_t(b_t)] \le (1-2\alpha g_t(\bid_t))$. Therefore, the relationship between $\lambda_{t+1}$ and $\lambda_t$ satisfies \[\lambda_{t+1}=\lambda_{t} \textrm{exp}[-\alpha g_t(b_t)] \le \lambda_{t} (1-2\alpha g_t(\bid_t))~.\]
Applying $g_t(\bid_t) \in[ -\frac{1}{\lambda_t}, 0)$, 
\begin{align*}
    \lambda_{t+1} - \lambda_{t} \le - 2\lambda_{t} \cdot \alpha g_t(\bid_t) \le 2\alpha~.
\end{align*}
Thus, Summarizing different cases, $\alpha^2g_t(b_t)^2 \cdot \max(\lambda_t, \lambda_{t+1})$ satisfies:
\[\frac{1}{2}\alpha^2g_t(\bid_t)^2 \cdot \max(\lambda_t, \lambda_{t+1}) \le \alpha (\lambda_t - \lambda_{t+1}) + 2\alpha^2.\]
Since $\alpha = \frac{1}{\sqrt{T}} \ge 1$ when $T\ge 1$, combining with Eqn.~\eqref{equation:alpha_g_lambda}, we have
\[
    \alpha g_t(\bid_t) \cdot \lambda_t \le 2(\lambda_t - \lambda_{t+1}) + 2\alpha^2~.
\]
Sum over $T$ rounds:
\begin{align*}
     \sum\nolimits_{t\in[T]} g_t(\bid_t) \cdot \lambda_t    & \le \frac{1}{\alpha} \sum\nolimits_{t\in[T]} \left(2(\lambda_t - \lambda_{t+1}) + 2\alpha^2 \right) \\
      & \le \frac{1}{\alpha} (4-2\lambda_{T+1}) \\
      & = O(\sqrt{T})~.
\end{align*}
The Eqn.~\eqref{equation:labmda_g_t} has been claimed. \qedhere
\end{proof}

%% file: apx-proof-fullfeedback.tex
\begin{lemma} \label{lem:full_feedback-distance_between_optimistic_CDF_and_CDF}
   For an unknown true CDF $\CDF$, we construct an estimated CDF $\empiricalCDF_m$ as in Eqn.~\eqref{eq:empiricalCDF} and define an optimistic CDF $\optimisticCDF_m$ as in Eqn.~\eqref{eq:optimisticCDF}. Consequently, the distance between $\optimisticCDF_m$ and $\CDF$ can be bounded by $2\eps_m$ and $\optimisticCDF_m(\bid) \ge \CDF(\bid)~$ for $\bid \in [0,1]$ with probability $1-O(\sfrac{1}{N_m})$, when $\eps_m$ satisfies $\eps_m  = \Theta\left(\frac{\log N_m}{\sqrt{N_m}}\right)$ ($N_m$ is the number of samples we used for estimation).
\end{lemma}
\begin{proof} [Proof of \Cref{lem:full_feedback-distance_between_optimistic_CDF_and_CDF}]
    Based on the definition of the estimated CDF $\empiricalCDF$, Dvoretzky–Kiefer–Wolfowitz inequality bounds the probability that the CDF $\empiricalCDF(\bid)$ differs from $\CDF(\bid)$ by more than a given $\eps_m > 0$:
    \[
      \prob{\sup_{\bid \in [0,1]} |\empiricalCDF_N(\bid)- \CDF(\bid)| > \eps_m} \le \exp(-2N_m\eps_m^2) ~.\]

    To ensure $\prob{\sup_{b \in [0,1]} |\empiricalCDF_{m}(\bid)- \CDF(\bid)| > \eps_{m}}$ is sufficiently small (i.e. approach zero as $N_m$ increases), $2N_m\eps_{m}^2$ must tend to infinity. Therefore, when $\eps_m$ satisfies \[
        \eps_m = \Theta\left(\frac{\log N_m}{\sqrt{N_m}}\right)~,
    \]
    the certainty of the closeness increases as $N_m$ increases. 
    Because $\CDF$ is within $\eps_m$ of $\empiricalCDF_m$ and $\optimisticCDF_m$ is defined as $\left(\empiricalCDF_m(\bid) + \eps_m \right ) \wedge 1$, it follows that
    \[\sup_\bid |\CDF(\bid) - \optimisticCDF_m(\bid)| \le 2\eps ~,\]
and
    \[\CDF(\bid) \le \empiricalCDF_m(\bid) + \eps = \optimisticCDF_m(\bid)~,\] 
for all $\bid$.
\end{proof}

\subsection{Analysis of \Cref{alg:randomized_strategy_under_F} with Estimated CDF} \label{sec:random_algorithm_with_F_analysis}
This section focuses on the analysis of \Cref{alg:randomized_strategy_under_F} with optimistic empirical CDF $\optimisticCDF$ relative to a true CDF $\CDF$. As a result, \Cref{lem:estimated_distribution_aganst_optimal_strategy} proved that near-optimal performance (Regret and ROI violation) is still guaranteed with the support of \Cref{lem:optimal_strategy_under_two_different_Fs} and \Cref{lem:algorithm_with_true_distribution_vs_estimated_distribution}.
\begin{lemma} \label{lem:estimated_distribution_aganst_optimal_strategy}
    Fix any value sequence $\valSeq^T$. Given the two CDFs, $\CDF$ and $\estimatedCDF$, such that $\estimatedCDF(b) \ge \CDF(b)$ for any $\bid \in [0,1]$ and $\sup_b |\CDF(\bid) - \estimatedCDF(\bid)| \le 2\eps$. 
    Run \Cref{alg:randomized_strategy_under_F} with the input CDF $\estimatedCDF$ (denoted by $\algWithF(\estimatedCDF)$) with $\CDF$ as environment for $T$ rounds. 
    The Regret satisfies:
    \[
        \TotalRegret(\algWithF(\estimatedCDF) \mid (\CDF, \valSeq^T)) \le O(\sqrt{T}) + T \cdot 2\eps~,
    \]
    and the ROI violation satisfies: 
    \[
        \TotalROIViolation(\algWithF(\estimatedCDF)\mid (\CDF, \valSeq^T)) \le 2\sqrt{T}\log T + T\cdot 2\eps~,
    \]    
\end{lemma}

\begin{proof}[Proof of \Cref{lem:estimated_distribution_aganst_optimal_strategy}]

From \Cref{lem:optimal_strategy_under_two_different_Fs}  and \Cref{lem:algorithm_with_true_distribution_vs_estimated_distribution}, we have
\[\TotalReward(\algWithF(\estimatedCDF)\mid (\CDF, \valSeq^T)) \ge \TotalReward(\algWithF(\estimatedCDF) \mid (\estimatedCDF, \valSeq^T)) - T\cdot 2\eps  \]
and 
\[\TotalReward( \opt(\estimatedCDF) \mid (\estimatedCDF, \valSeq^T)) \ge \TotalReward (\opt(\CDF) \mid (\CDF, \valSeq^T))~. \]
Here, we denote output of the optimal bidding algorithm in hindsight with $\estimatedCDF$ as input by $\opt(\estimatedCDF)$. Similarly, we denote the optimal bidding algorithm in hindsight with $\CDF$ as input by $\opt(\CDF)$.

Since the Regret of \Cref{alg:randomized_strategy_under_F} with known true CDF is $O(\sqrt{T})$, we can infer that
\[\TotalReward(\algWithF(\estimatedCDF) \mid (\estimatedCDF, \valSeq^T)) = \TotalReward(\opt(\estimatedCDF) \mid (\estimatedCDF, \valSeq^T)) - O(\sqrt{T})~.\]
Combining above, we obtain
\begin{align*}
    \TotalReward(\algWithF(\estimatedCDF) \mid (\CDF, \valSeq^T))  & \ge \TotalReward(\algWithF(\estimatedCDF) \mid (\estimatedCDF, \valSeq^T)) - T\cdot 2\eps \\ 
    & = \TotalReward(\opt(\estimatedCDF) \mid (\estimatedCDF, \valSeq^T)) - T\cdot 2\eps - O(\sqrt{T}) \\ 
    & \ge \TotalReward(\opt(\CDF) \mid (\CDF, \valSeq^T)) - T\cdot 2\eps - O(\sqrt{T})~.
\end{align*}
Therefore, the Regret satisfies
\begin{align*}
    \TotalRegret(\algWithF(\estimatedCDF) \mid (\CDF, \valSeq^T))  & =  \TotalReward(\opt(\CDF) \mid (\CDF, \valSeq^T)) - \TotalReward(\algWithF(\estimatedCDF) \mid (\CDF, \valSeq^T)) \\ 
    & = O(\sqrt{T}) +  T\cdot 2\eps~.
\end{align*}
For ROI violation, combining \[\TotalROIViolation(\algWithF(\estimatedCDF) \mid (\CDF, \valSeq^T)) \le \TotalROIViolation(\algWithF(\estimatedCDF) \mid (\estimatedCDF, \valSeq^T)) + T\cdot 2\eps \]
and
\[\TotalROIViolation(\algWithF(\estimatedCDF) \mid ( \estimatedCDF, \valSeq^T)) \le 2\sqrt{T}\log T~,\]
We have \[
    \TotalROIViolation(\algWithF(\estimatedCDF) \mid (\CDF, \valSeq^T)) \le 2\sqrt{T}\log T + T\cdot 2\eps~. \qedhere
\] 
\end{proof}


\begin{lemma}\label{lem:optimal_strategy_under_two_different_Fs}
Fix any value sequence $\valSeq^T$. Given $\estimatedCDF$ and $\CDF$ are CDFs of maximum competing bid, when $\estimatedCDF(\bid) \ge \CDF(\bid)$ for any $\bid \in [0,1]$, we claim that this implies the Reward of the optimal bidding algorithm in hindsight with CDF input $\estimatedCDF$ (denoted by $\opt(\estimatedCDF)$) and $\CDF$ (denoted by $\opt(\CDF)$) satisfies: 
        \[
           \TotalReward( \opt(\estimatedCDF) \mid (\estimatedCDF, \valSeq^T)) \ge \TotalReward (\opt(\CDF) \mid (\CDF, \valSeq^T))~.
        \]
\end{lemma}
\begin{proof}[Proof of \Cref{lem:optimal_strategy_under_two_different_Fs}]
    Construct the concave envelope $\allocPayConv$ and $\allocPayConv^\dag$:
    \begin{itemize}[topsep=2pt,itemsep=1pt,leftmargin=2.5ex,parsep=1pt]
        \item Given $\CDF(\bid)$, define a curve $\allocPay$ as 
          \[
              y = \allocPay(x) , y = \CDF(\bid), x=b \cdot \CDF(\bid)~.
         \] 
        Let $\allocPayConv$ denote the concave envelope of $\allocPay$.  
        \item Similarly, for $\estimatedCDF(\bid)$, define a curve $\allocPay^\dag$, where 
        \[
            y = \allocPay^\dag(x) , y =\estimatedCDF(\bid), x=b \cdot \estimatedCDF(\bid)~,
        \] 
        and denote its concave envelope as $\allocPayConv^\dag$.
    \end{itemize}
    Due to the property of the cumulative distribution function, $\CDF(\bid)$ and $\estimatedCDF(\bid)$ both have definition at $b=0$ and $\estimatedCDF(0) \ge \CDF(0) \ge 0$.\\
    Fix a bit $b$, there exists a point $(x, y)$ the curve $\allocPay$: 
    \[x=b\cdot \CDF(\bid), y = \CDF(\bid)~,\] 
    and a corresponding point  $(x^\dag, y^\dag)$ on the curve $\allocPay^\dag$: 
    \[x^\dag=b \cdot \estimatedCDF(\bid), y^\dag= \estimatedCDF(\bid)~.\]
    Notice that both points lie on the line: $y=\frac{1}{b}x$. Because $\estimatedCDF(\bid) \ge \CDF(\bid)$,  $\frac{\CDF(\bid)}{\estimatedCDF(\bid)} \in [0,1]$. Observe that the point $(x,y)$ is a convex combination of $(x^\dag, y^\dag)$ and $(0,0)$:
    \[x=\frac{\CDF(\bid)}{\estimatedCDF(\bid)} \cdot x^\dag + \left(1-\frac{\CDF(\bid)}{\estimatedCDF(\bid)}\right)\cdot 0 ~,\]
    \[y=\frac{\CDF(\bid)}{\estimatedCDF(\bid)} \cdot y^\dag + \left(1-\frac{\CDF(\bid)}{\estimatedCDF(\bid)}\right)\cdot 0 ~.\]
    This shows that every point on the curve $\allocPay$ is contained in the convex hull of $\allocPay^\dag \cup \{(0,0)\}$. Due to the property of CDF, $\CDF(\bid)$ and $\estimatedCDF(\bid)$ both have definition at $b=0$ and $\estimatedCDF(0) \ge \CDF(0) \ge 0$. $(0,0)$ is in both the convex hull of $\allocPay$ and the convex hull of $\allocPay^\dag$. Thus, the convex hull of $\allocPay$ is a subset of the convex hull of $\allocPay^\dag$. 

    By the definition of concave envelope, the lowest-valued concave function that overestimates or equals the original function over that set, it follows
    \[\allocPayConv \le \allocPayConv^\dag ~,\] 
    pointwise. \\
    Thus, for a fixed value sequence $\valSeq_T$, given CDF $\CDF$, an optimal bid sequence $\{\bid_t\}_{t=1}^{T}$ will be generated. For each $\bid_t$ at round $t$, we could find $\bid^\dag_t$ under $\CDF^\dag$ with the same payment ($\bid_t\CDF(\bid_t) = \bid^\dag_t \CDF(\bid^\dag_t)$) but a potentially higher value ($\val_t \cdot \CDF^\dag(\bid^\dag_t) \ge \val_t \cdot \CDF(\bid_t)$). Therefore, the reward satisfies
    \begin{align*}
       \TotalReward ( \opt(\CDF) \mid (\CDF, \valSeq^T))
        & = \sum\nolimits_{t\in[T]} v_t \cdot \CDF(\bid_t) \\
        & \le \sum\nolimits_{t\in[T]} v_t \cdot \estimatedCDF(\bid^\dag_t) \\
        & = \TotalReward( \opt(\estimatedCDF) \mid (\estimatedCDF, \valSeq^T)) ~.  \qedhere  
    \end{align*}
\end{proof}


\begin{lemma} \label{lem:algorithm_with_true_distribution_vs_estimated_distribution}
    Fix a value sequence $\valSeq^T$. Given the two CDFs, $\CDF$ and $\estimatedCDF$, such that $\estimatedCDF(b) \ge \CDF(b)$ and $\sup_b |\CDF(\bid) - \estimatedCDF(\bid)| \le 2\eps$. Run \Cref{alg:randomized_strategy_under_F} with $\estimatedCDF$ as its input, denoted by $\algWithF(\estimatedCDF)$. The reward satisfies 
        \[\TotalReward(\algWithF(\estimatedCDF) \mid (\CDF, \valSeq^T)) \geq \TotalReward(\algWithF(\estimatedCDF) \mid (\estimatedCDF, \valSeq^T)) - T \cdot 2 \epsilon ~,\]
    and the ROI violation cannot be amplified by more than $T \cdot 2\eps $: 
        \[\TotalROIViolation(\algWithF(\estimatedCDF) \mid (\CDF, \valSeq^T)) \le \TotalROIViolation(\algWithF(\estimatedCDF) \mid (\estimatedCDF, \valSeq^T)) + T \cdot 2 \epsilon ~.\]
\end{lemma}

\begin{proof}[Proof of \Cref{lem:algorithm_with_true_distribution_vs_estimated_distribution}]
   At round $t$, given the input value $\val_t$, the bidding algorithm $\algWithF(\estimatedCDF)$ generates $\bidConv^\dag$ and the randomized strategy $\bid^\dag_1$ with probability $p_1$ and $\bid^\dag_2$ with probability $p_2$. $p_1 + p_2 = 1$ and $p_1 \ge 0, p_2 \ge 0$.
   
    Under $\estimatedCDF$, it generates expected reward 
    \[\reward^\dag_t = \val_t \cdot [p_1 \cdot \estimatedCDF(\bid^\dag_1) + p_2 \cdot \estimatedCDF(\bid^\dag_2)] ~, \] 
    and expected payment 
    \[\payment^\dag_t = p_1 \cdot \bid^\dag_1 \cdot \estimatedCDF(\bid^\dag_1) + p_2 \cdot \bid^\dag_2 \cdot \estimatedCDF(\bid^\dag_2) ~.\]
   When the distribution is $\CDF$, it generates expected reward 
    \[\reward_t = v_t \cdot [p_1 \cdot \CDF(\bid^\dag_1) + p_2 \cdot F(\bid^\dag_2)] ~,\] 
    and expected payment 
    \[\payment_t = p_1 \cdot \bid^\dag_1 \cdot \CDF(\bid^\dag_1) + p_2 \cdot \bid^\dag_2 \cdot \CDF(\bid^\dag_2) ~.\] 
Because $\sup_b |\CDF(\bid) - \estimatedCDF(\bid)| \le 2\eps$ and $\CDF(\bid) \le \estimatedCDF(\bid)$, The difference in value will be $\CDF(\bid) \ge \estimatedCDF(\bid) - 2 \eps $.

The difference of expected reward will be 
\begin{align*}
  \reward_t - \reward^\dag_t  & = \val_t \cdot \{p_1 \cdot [F(\bid^\dag_1) - \estimatedCDF (b^\dag_{t_2})] + p_2 \cdot [F(\bid^\dag_1) - \estimatedCDF (b^\dag_{t_2})]\}  \\
   & \ge \val_t \cdot [ p_1 \cdot (-2 \eps) +   p_2 \cdot (-2 \eps)] \\
   & \ge -2\eps ~. 
\end{align*}
The difference of expected payment will be 
\begin{align*}
\payment_t - \payment^\dag_t = p_1 \cdot \bid^\dag_1 \cdot [\CDF(\bid^\dag_1) - \estimatedCDF(\bid^\dag_1)] + p_2 \cdot \bid^\dag_2 \cdot [F(\bid^\dag_2) - \estimatedCDF(\bid^\dag_2)] \le 0 ~.
\end{align*}
Sum of $T$ rounds, the difference of expected reward satisfies:
\[
    \sum\nolimits_{t\in[T]}  \reward_t - \sum\nolimits_{t\in[T]} \reward^\dag_t \ge -T \cdot 2\eps ~,\]
and the difference of expected payment satisfies:
\[
    \sum\nolimits_{t\in[T]}  \payment_t - \sum\nolimits_{t\in[T]} \payment^\dag_t  \le 0 ~.
\]
Therefore, the over all reward satisfies:
\begin{align*}
    \TotalReward(\algWithF(\estimatedCDF) \mid (\CDF, \valSeq^T))  & = \sum\nolimits_{t\in[T]}  \reward_t \\
     & \ge \sum\nolimits_{t\in[T]} \reward^\dag_t -T \cdot 2\eps \\ 
     & = \TotalReward(\algWithF(\estimatedCDF) \mid (\estimatedCDF, \valSeq^T)) - T \cdot 2 \epsilon ~.
\end{align*}

The ROI violation satisfies:
\begin{align*}
    \TotalROIViolation(\algWithF(\estimatedCDF) \mid (\CDF,  \valSeq^T))  & = \TotalPayment \left(\algWithF(\estimatedCDF) \mid (\CDF,  \valSeq^T) \right) - \TotalReward \left(\algWithF(\estimatedCDF) \mid (\CDF,  \valSeq^T) \right) \\
     & \le \TotalPayment \left(\algWithF(\estimatedCDF) \mid (\estimatedCDF, \valSeq^T) \right) - \TotalReward \left( \algWithF(\estimatedCDF) \mid (\estimatedCDF, \valSeq^T) \right) + T \cdot 2 \eps  \\
     & = \TotalROIViolation \left( \algWithF(\estimatedCDF) \mid (\estimatedCDF, \valSeq^T) \right) + T \cdot 2 \eps ~. \qedhere
\end{align*}
\end{proof}

\subsection{Analysis of Regret and ROI Violation of \Cref{alg:randomized_strategy}}
In this section, we analyze the performance of \Cref{alg:randomized_strategy} in terms of regret and ROI violation with $\eps_{m} = \frac{\log{N_m}}{\sqrt{N_m}}$ for each stage $m$.
\begin{lemma}\label{lem:full_feedback-Regret}
Consider an unknown CDF $\CDF$ of maximum competing bid. Given a fixed value sequence $\valSeq^T$, run \Cref{alg:randomized_strategy} for $T$ rounds. With $\eps_{m} = \frac{\log{N_m}}{\sqrt{N_m}}$ for each stage $m$, the reward satisfies
\begin{equation}
\begin{aligned}
  \TotalRegret(\mathrm{\Cref{alg:randomized_strategy}} \mid (\CDF, \valSeq^T)) = \widetilde O(\sqrt T) ~. 
\end{aligned}
\end{equation}
\end{lemma}
\begin{proof}[Proof of \Cref{lem:full_feedback-Regret}]
For each stage $m$, we run \Cref{alg:randomized_strategy_under_F} with the optimistic CDF $\optimisticCDF_m$ for $T_m$ rounds, where $N_m=2^{m-1}-1$ and $T_m=2^{m-1}$.

Based on \Cref{lem:estimated_distribution_aganst_optimal_strategy} and $\eps_{m} = \frac{\log{N_m}}{\sqrt{N_m}}$, we have 
\begin{equation}
\begin{aligned}
  \TotalReward(\algWithF(\optimisticCDF_m) \mid ( \CDF, \valSeq_m))  & \ge \TotalReward(\opt \mid (\CDF, \valSeq_m)) - T_m \cdot 2 \eps_m  \\
   & = \TotalReward(\opt \mid (\CDF, \valSeq_m)) - \widetilde O(\sqrt{T_m})~,
\end{aligned}
\end{equation}
where $\valSeq_m$ is the sub-sequence of values at stage $m$ including $T_m$ elements, and $\algWithF(\optimisticCDF_m)$ the bids generated by feeding \Cref{alg:randomized_strategy_under_F} with a distribution $\optimisticCDF_m$.

The total reward across $M$ stages is
\begin{align*}
  \TotalReward(\mathrm{\Cref{alg:randomized_strategy}} \mid (\CDF, \valSeq^T))  & = \sum\nolimits_{m\in[M]} \TotalReward(\algWithF(\optimisticCDF_m) \mid ( \CDF, \valSeq_m))  \\
   & \ge  \TotalReward(\opt(\CDF) \mid (\CDF, \valSeq^T)) - \sum\nolimits_{m\in[M]} \sqrt{2^{m-1}}\\
   & = \TotalReward(\opt(\CDF) \mid (\CDF, \valSeq^T)) - \widetilde O(\sqrt{T})~. \qedhere
\end{align*} 
\end{proof}

\begin{lemma}\label{lem:full_feedback-ROI}
    Consider an unknown CDF $\CDF$ of maximum competing bid. Given a fixed value sequence $\valSeq^T$, run \Cref{alg:randomized_strategy} for $T$ rounds. With $\eps_{m} = \frac{\log{N_m}}{\sqrt{N_m}}$ for each stage $m$, the ROI Violation satisfies: 
    \begin{equation}
    \begin{aligned}
        \TotalROIViolation( \mathrm{\Cref{alg:randomized_strategy}} \mid (\CDF, \valSeq^T)) = \widetilde O(\sqrt{T})~. 
    \end{aligned}
    \end{equation}
\end{lemma}
\begin{proof}[Proof of \Cref{lem:full_feedback-ROI}]
For each stage $m$, we run \Cref{alg:randomized_strategy_under_F} with the optimistic CDF $\optimisticCDF_m$ for $T_m=2^{m-1}$ rounds, where $N_m=2^{m-1}-1$ is the number of samples collected in the preceding stages.

Based on \Cref{lem:estimated_distribution_aganst_optimal_strategy} and $\eps_m = \frac{\log{N_m}}{\sqrt{N_m}}$, The ROI violation is bounded by
\begin{equation}
\begin{aligned}
  \TotalROIViolation(\algWithF(\optimisticCDF_m) \mid (\CDF, \valSeq_m))  & \le 2\sqrt{T_m} \log T_m + T_m \cdot 2 \eps_m  \\
   & \le 2\sqrt{T_m} \cdot (\log T_m + 1)~.
\end{aligned}
\end{equation}
The total ROI violation across $M=\lceil \log_2(T+1) \rceil$ stages is bounded by
\begin{align*}
  \TotalROIViolation(\mathrm{\Cref{alg:randomized_strategy}} \mid (\CDF, \valSeq^T))  & = \sum\nolimits_{m\in[M]} \TotalROIViolation(\algWithF(\optimisticCDF_m) \mid (\CDF, \valSeq_m))  \\
   & \le \sum\nolimits_{m\in[M]} 2 \sqrt{2^{m-1}} (\log (2^{m-1}) + 1)\\
   & \le \sum\nolimits_{m\in[M]} 2\log 2 \cdot (m-1) \cdot \sqrt{2}^{m-1}  + \sum\nolimits_{m\in[M]} 2\sqrt{2}^{m-1} \\
   & = O(M\cdot \sqrt{2^M}) \\
   & = O(\sqrt{T}\log T)~.
   \qedhere
\end{align*}
\end{proof}

%% file: apx-proof-banditfeedback.tex
\subsection{Construct the Empirical Distribution and its Optimistic Distribution} \label{sec:bandit_feedback_contruct_empirical_distribution}

\begin{lemma} 
\label{lem:bandit_feedback-CDF_bound}
    Given the grid points $\{c_k\}_{k\in[K-1]\cup\{0\}}$ where we bid every point consecutively for $M$ time rounds, the probability that there exists one grid bid $c_k$ for which the absolute difference between the estimated CDF $\empiricalCDF(c_k)$ and the true CDF $F(c_k)$ exceeds $\varepsilon > 0$ is bounded by
    \[
        \prob{\exists k \in [K]: |\empiricalCDF(c_k) - \CDF(c_k)| \ge \varepsilon} 
        \le K \cdot \exp(-2M\eps^2)~.
    \]
\end{lemma}
\begin{proof}[Proof of \Cref{lem:bandit_feedback-CDF_bound}]
We define $\widetilde \CDF(c_k) \triangleq \frac{1}{M} \cdot \sum\nolimits_{t\in [k\cdot M+1: (k+1)\cdot M]} \indicator{\bid_t \ge \competingbid_t}$, for notation simplicity. By Chernoff bound, for each grid point $c_k$ and a given $\varepsilon > 0$, we have
    \[
        \prob{|\widetilde \CDF(c_k) - \CDF(c_k)| \ge \varepsilon} \le \exp(-2 M \varepsilon^2)~.
    \]
    Now consider $|\empiricalCDF(c_k) - \CDF(c_k)|$ where $i \le k$. Since $\widetilde \CDF(c_i) \le \empiricalCDF(c_k)$ and $F(c_i) \le  \CDF(c_k)$ for all $i \le k$, we have:
    \begin{align*}
        \empiricalCDF(c_k) -  \CDF(c_k) &= \max\nolimits_{i\in [k]\cup\{0\}} \{\widetilde  \CDF(c_k)\} -  \CDF(c_k) \\
        &= \max\nolimits_{i\in [k]\cup\{0\}} \{\widetilde  \CDF(c_i) -  \CDF(c_i) +  \CDF(c_i)\} -  \CDF(c_k) \\
        &\le \max\nolimits_{i\in [k]\cup\{0\}} \{\widetilde  \CDF(c_i) -  \CDF(c_i)\} + \max\nolimits_{i\in [k]\cup\{0\}} \{F(c_i)\} -  \CDF(c_k) \\
        &\le \max\nolimits_{i\in [k]\cup\{0\}} \{\widetilde  \CDF(c_i) -  \CDF(c_i)\} ~.
    \end{align*}
    Similarly,
    \begin{align*}
         \CDF(c_k) - \empiricalCDF(c_k) 
         & =  \CDF(c_i) - \max\nolimits_{i\in [k]\cup\{0\}} \{\widetilde  \CDF(c_i)\} \\ 
         & \le  \CDF(c_k) - \widetilde  \CDF(c_k)\\
        & \le |\CDF(c_k) - \widetilde  \CDF(c_k)| ~.
    \end{align*}
    Therefore, $|\empiricalCDF(c_k) - \CDF(c_k)| \le \max\nolimits_{i\in [k]\cup\{0\}} |\widetilde  \CDF(c_i) -  \CDF(c_i)|$. If $|\widetilde  \CDF(c_i) -  \CDF(c_i)| < \varepsilon$ for all $i \le k$, then $|\empiricalCDF(c_k) -  \CDF(c_k)| < \varepsilon$.

    Thus, the event $|\empiricalCDF(c_k) - \CDF(c_k)| \ge \varepsilon$ implies that there exists at least one $i \le k$ such that $|\widetilde  \CDF(c_i) -  \CDF(c_i)| \ge \varepsilon$. $\widetilde \CDF(c_i)$ and $\CDF(c_k)$ is also bounded by
    \[
            \prob{|\empiricalCDF(c_k) - \CDF(c_k)| \ge \varepsilon} \le \exp(-2 M \varepsilon^2) ~.
    \]
    
    By the union bound, 
    \[
        \prob{\exists i \in [K]: |\empiricalCDF(c_k) -  \CDF(c_k)| \ge \eps} 
        \le \sum_i \prob{|\empiricalCDF(c_k) -  \CDF(c_k)| \ge \eps} 
        \le K \cdot \exp(-2M\eps^2) ~. \qedhere
    \]
\end{proof}

\begin{lemma} \label{lem:bandit_feedback_F_vs_bar_F_prime}
    Given the grid points $\{c_k\}_{k\in[K]\cup\{0\}}$, and an empirical CDF $\optimisticCDF$ constructed in such that $\optimisticCDF(c_k) \ge \CDF(c_k)$ for all $k \in [K]\cup\{0\} $), if we construct another CDF $\conservCDF$ as:
    \[ \conservCDF(\bid) = \optimisticCDF((\bid + \sfrac{1}{K}) \wedge 1)~, \]
    then $\conservCDF(\bid) \ge \CDF (\bid), \bid \in [0,1]~$.
\end{lemma}
\begin{proof}[Proof of \Cref{lem:bandit_feedback_F_vs_bar_F_prime}]
    Given that $\optimisticCDF(c_k) \ge \CDF(c_k)$ for all 
    $k \in \{0,1,\ldots, K\}$.
    Consider an interval $[c_k, c_{k+1})$. For any $b \in [c_k, c_{k+1})$, by the definition of $\conservCDF(b)$, we have
        \[\conservCDF (b) = \optimisticCDF(\min(b+\sfrac{1}{K}, 1)) ~.\]
    Since $b \in [c_k, c_{k+1})$, we have $b + \sfrac{1}{K} \in [c_{k+1}, c_{k+2})$. Since $\optimisticCDF$ is a step function constant in the range $[c_{k+1},c_{k+2})$, we have
        \[ \optimisticCDF(\min(b+\sfrac{1}{K}, 1)) = \optimisticCDF(\min(c_{k+1}, 1))~.\]
    Both $\optimisticCDF$ and $\CDF$ are non-decreasing functions, we have
        \[\optimisticCDF(\min(c_{k+1}, 1)) \ge \CDF(c_{k+1}) \ge \CDF(b)~.\]
    Combining these inequalities, we get:
        \[\conservCDF(b) \ge \CDF(b), b \in [c_k, c_{k+1})~. \]
    This holds for all intervals defined by the grid points. Therefore,
       \[ \conservCDF(\bid) \ge \CDF (\bid), \bid \in [0,1]~. \qedhere\]
\end{proof}

\subsection{Implement \Cref{alg:randomized_strategy_under_F} Conservatively with Constructed $\optimisticCDF$ in Exploration Phase} \label{sec:bandit_feedback_randomized_strategy_with_F_explore_phase}
\label{sec:conservative_strategy}

Given the optimistic function $\optimisticCDF$, which is the estimated CDF of the maximum competing bid defined in Eqn.~\eqref{eq:optimisticCDF one-bit}. The allocation-payment curve $\optAllocPay$  defined in \Cref{defn:alloc-pay curve} is constructed as 
$\optAllocPay(\bid \cdot \optimisticCDF(\bid)) = \optimisticCDF(\bid)$ for all $b \in [0, 1]$.

For each interval $\bid \in [c_k, c_{k+1})$, $x$ increases linearly with $\bid$ from $c_k \cdot \optimisticCDF(c_k)$ to $ c_{k+1} \cdot \optimisticCDF(c_k)$, while $y$ remains constant at $\optimisticCDF(c_k$.  This leads to a "staircase" shape for the $\optAllocPay$ curve when considering all intervals.

The algorithm then computes the concave envelope of this $\optAllocPay$ curve. Due to the attributes of the concave envelope, the concave envelope of $\optAllocPay$ is formed by the linear segments connecting a subset of "breakpoints": $(c_{k} \cdot \optimisticCDF(c_k), \optimisticCDF(c_k))$ for $k = 0,1,\ldots, K$.\\
Specifically, for a segment between to consecutive breakpoints $(c_k \cdot \optimisticCDF(c_k), \optimisticCDF(c_k)), (c_{k+1} \cdot \optimisticCDF(c_{k+1}), \optimisticCDF(c_{k+1}))$, the linear function is given by
\begin{align*}
   & \yConv= m_i (\xConv - c_k \cdot \optimisticCDF(c_k)) +  \optimisticCDF(c_k)~,
\end{align*} 
where the slope $m_k$ will be
\begin{align*}
   & m_k=\frac{\optimisticCDF(c_{k+1})- \optimisticCDF(c_{k})}{c_{k+1} \cdot \optimisticCDF(c_{k+1}) - c_{k} \cdot \optimisticCDF(c_{k})} ~.
\end{align*} 

The algorithm iteratively constructs the upper concave envelope by checking the slopes. If $m_i < m_{i-1}$, it indicates a concavity violation at the intermediate point $(c_k \cdot \optimisticCDF(c_k), \optimisticCDF(c_k))$, and this point is bypassed by forming a new linear segment between $(c_{k-1} \cdot \optimisticCDF(c_{k-1}), \optimisticCDF(c_{k-1}))$ and $(c_{k+1} \cdot \optimisticCDF(c_{k+1}))$.

When $\optimisticCDF$ is used as input of \Cref{alg:randomized_strategy_under_F}, the algorithm will effectively compute this concave envelope and the corresponding randomized strategy, which will be a probability distribution over the bids corresponding to the breakpoints of this envelope. These bids are precisely the subset of the grid points $\{c_k\}_{k=0}^{K}$ used in the exploration phase.


Given a bid $b$ generated by \Cref{alg:randomized_strategy_under_F}, the conservative strategy, denoted by $\ConsAlg$, submits a bid as $\min(b + \sfrac{1}{K}, 1)$.

\begin{lemma} \label{lem:conservative_strategy_vs_randomized_algorithm}
    Given the the optimistic CDF $\optimisticCDF$ (defined in Eqn.~\eqref{eq:optimisticCDF one-bit}) and its conservative version $\conservCDF$ (defined in Eqn.~\eqref{eq:conservCDF}). Let $\algWithF(\conservCDF)$ denote the randomized strategy generated by \Cref{alg:randomized_strategy_under_F} with input $\conservCDF$, and $\ConsAlg(\conservCDF)$ be the corresponding conservative strategy. Fix a value sequence $\valSeq^T$ and run the conservative strategy $\ConsAlg$ for $T$ rounds. The total reward of the conservative strategy under the true distribution $\optimisticCDF$ satisfies:
    \[\TotalReward(\ConsAlg(\conservCDF) \mid (\optimisticCDF, \valSeq^T)) = \TotalReward(\algWithF(\conservCDF) \mid (\conservCDF, \valSeq^T)).\]
    The ROI violation satisfies \[\TotalROIViolation(\ConsAlg(\conservCDF) \mid ( \optimisticCDF, \valSeq^T)) \le \TotalROIViolation(\algWithF(\conservCDF) \mid ( \conservCDF, \valSeq^T)) + \frac{T}{K}.\]
 
\end{lemma}

\begin{proof}[Proof of \Cref{lem:conservative_strategy_vs_randomized_algorithm}]
  Let $\bid_t$ be the bid generated by $\Algorithm(\conservCDF)$ at round $t$, and $\conservBid_t=\min(\bid_t+\sfrac{1}{K}, 1)$ is the bid of $\ConsAlg(\conservCDF)$. 
  By the definition of the conservative strategy $\ConsAlg$, the output of $\ConsAlg$
  \[
    \conservBid=\min(\bid+\sfrac{1}{K}, 1) ~.
  \]
  By the definition of $\conservCDF(\bid_t)$
  \[
    \conservCDF(\bid_t) =\optimisticCDF(\min(\bid_t+\sfrac{1}{K}, 1)) = \optimisticCDF(\conservBid_t) ~, 
  \]
  thus, we can derive that 
  \begin{align*}
     \TotalReward(\ConsAlg(\conservCDF) \mid (\optimisticCDF, \valSeq^T))  & = \sum\nolimits_{t\in[T]} \val_t \cdot \optimisticCDF(\conservBid_t)\\ 
       & =  \sum\nolimits_{t\in[T]} \val_t \cdot \conservCDF(\bid_t) \\
        & = \TotalReward(\algWithF(\conservCDF) \mid (\conservCDF, \valSeq^T)) ~. 
  \end{align*}
  For the payment,
  \begin{align*}
     \TotalPayment(\ConsAlg(\conservCDF) \mid (\optimisticCDF, \valSeq^T))  & = \sum\nolimits_{t\in[T]} \conservBid_t \cdot \optimisticCDF(\conservBid_t)\\ 
       & =  \sum\nolimits_{t\in[T]} \min(\bid_t + \sfrac{1}{K}, 1) \cdot \conservCDF(b_t) \\
        & \le \TotalPayment(\algWithF(\conservCDF) \mid (\conservCDF, \valSeq^T)) + \frac{T}{K} ~. 
  \end{align*}
Thus, the ROI violation satisfies
  \begin{align*}
     \TotalROIViolation(\ConsAlg(\conservCDF) \mid (\optimisticCDF, \valSeq^T))  & = \TotalReward(\ConsAlg(\conservCDF) \mid (\optimisticCDF, \valSeq^T)) - \TotalPayment(\ConsAlg(\conservCDF) \mid (\optimisticCDF, \valSeq^T)) \\ 
       & \le \TotalReward(\algWithF(\conservCDF) \mid (\conservCDF, \valSeq^T)) - \TotalPayment(\algWithF(\conservCDF) \mid (\conservCDF, \valSeq^T)) + \frac{T}{K} \\
       & \le \TotalROIViolation(\algWithF(\conservCDF) \mid (\conservCDF, \valSeq^T)) + \frac{T}{K} ~. \qedhere
  \end{align*}
\end{proof}

\subsection{Analyze the Performance of Implementing \Cref{alg:randomized_strategy_under_F} Conservatively}
\label{sec:bandit_feedback_analysis_conservative_strategy}

The optimal randomized strategy algorithm ( \Cref{alg:randomized_strategy_under_F}), discussed in \Cref{sec:optimal_random} (with a known true distribution $\CDF$), achieves a regret of $\sqrt{T}$ and a ROI violation of $O(\sqrt{T}\cdot \log T)$ relative to the optimal randomized strategy subject to the strict ROI constraint under $\CDF$.

In \Cref{sec:bandit_feedback_randomized_strategy_with_F_explore_phase}, we detailed how the randomized bidding strategy is designed using the optimistic distribution $\optimisticCDF$ constructed during the exploration phase. The bids submitted are drawn from the set of bids explored.

This section provides a theoretical analysis of the conservative strategy based on the output of \Cref{alg:randomized_strategy_under_F}, which was introduced in \Cref{sec:conservative_strategy}. We aim to understand its performance in terms of reward and ROI violation when the underlying distribution of competing bids is unknown and we rely on our estimate $\conservCDF$.


\begin{lemma} \label{lem:bandit_feedback_algorithm_with_true_distribution_vs_estimated_distribution}
    Consider a fixed value sequence $\valSeq^T$ and a true CDF $\CDF$ of the maximum competing bid. The optimistic CDF $\optimisticCDF$ (Eqn.~\eqref{eq:optimisticCDF one-bit}) and its conservative variant $\conservCDF$ (Eqn.~\eqref{eq:conservCDF})are obtained in the exploration phase of $\Cref{alg:randomized_strategy_with_bandit_feedback}$ using a grid point set $\{c_k\}_{k=0}^{K} (c_k=\sfrac{k}{K})$.  For points in the point set, it satisfies $\optimisticCDF(c_k) \ge \CDF(c_k)$ and $\optimisticCDF(c_k) - \CDF(c_k) \le 2\varepsilon,  \forall c_k \in \{c_k\}_{k=0}^{K}$.
    Let $\ConsAlg(\conservCDF)$ denote the conservative strategy based on \Cref{alg:randomized_strategy_under_F} with input $\conservCDF$. Then, Running $\ConsAlg(\conservCDF)$ for $T$ rounds, the total reward over the true distribution $\CDF$ satisfies: 
        \[\TotalReward(\ConsAlg(\conservCDF) \mid (\CDF, \valSeq^T)) \geq \TotalReward(\ConsAlg(\conservCDF) \mid (\optimisticCDF, \valSeq^T)) - T \cdot 2 \epsilon ~,\]
    and the ROI violation cannot be amplified by more than $T \cdot 2\eps~$. 
        \[\Delta(\ConsAlg(\conservCDF) \mid (\CDF, \valSeq^T)) \leq \Delta(\ConsAlg(\conservCDF) \mid (\optimisticCDF, \valSeq^T)) + T \cdot 2 \epsilon ~.\]

\end{lemma}
\begin{proof}[Proof of \Cref{lem:bandit_feedback_algorithm_with_true_distribution_vs_estimated_distribution}]
   Consider a single round $t$, where the input value is $\val_t$. \Cref{alg:randomized_strategy_under_F} $\Algorithm(\conservCDF)$ generates an optimal bid $\bidConv$, which, as analyzed in \Cref{thm:optimal_randomized_strategy}, obtains a randomized strategy $\bid_{1}$ with probability $p_1$ and $\bid_{2}$ with probability $p_2$. $p_1 + p_2 = 1$ and $p_1 \ge 0, p_2 \ge 0$. \Cref{sec:bandit_feedback_randomized_strategy_with_F_explore_phase} showed that $\bid_{1}$ and $\bid_{2}$ are in the grid points $\{c_k\}_{k=0}^{K} (c_k=\sfrac{k}{K})$. Further, the conservative strategy $\ConsAlg$ will generate $\conservBid_{1}=\min(b_1+\sfrac{1}{K}, 1), \conservBid_{2}=\min(b_2+\sfrac{1}{K},1)$, which are still in the point set $\{c_k\}_{k=0}^{K}$. \\
    Under $\optimisticCDF$, the expected reward for the conservative strategy in this round, denoted by $\bar \reward_t$, satisfies:
    \[
        \bar \reward_t= \val_t \cdot [p_1 \cdot \optimisticCDF(\conservBid_{1}) + p_2 \cdot \optimisticCDF(\conservBid_2)]~.
    \] 
    The expected payment, denoted by $\bar \payment_t$, satisfies:
    \[
        \bar \payment_t = p_1 \cdot \conservBid_1 \cdot \optimisticCDF(\conservBid_1) + p_2 \cdot \conservBid_2 \cdot \optimisticCDF(\conservBid_2)~.
    \] 
    Similarly, under $\CDF$, the expected reward, denoted by $\reward_t$, satisfies:
    \[
        \reward_t = \val_t \cdot [p_1 \cdot F(\conservBid_1) + p_2 \cdot F(\conservBid_2)]~.
    \] 
    The expected payment, denoted by $\payment_t$, satisfies:
    \[
        \payment_t = p_1 \cdot \conservBid_1 \cdot F(\conservBid_1) + p_2 \cdot \conservBid_2 \cdot \CDF(\conservBid_2)~.
    \] 
    Given $|\optimisticCDF(b) - \CDF(b)| \le 2\varepsilon$ and $\CDF(\bid) \le \optimisticCDF(\bid)$, for any $ \bid \in \{c_k\}_{k=0}^{K} $, the difference in value would be $\CDF(\bid) \ge \optimisticCDF - 2 \eps~$.
Difference in expected reward in round $t$ satisfies:
\begin{align*}
  \reward_t - \bar \reward_t  & = \val_t \cdot \{p_1 \cdot [\CDF (\conservBid_1) - \optimisticCDF (\conservBid_2)] + p_2 \cdot [\CDF(\conservBid_1) - \optimisticCDF (\conservBid_2)]\}  \\
   & \ge \val_t \cdot [ p_1 \cdot (-2 \eps) +   p_2 \cdot (-2 \eps)] \\
   & \ge -2\eps ~. 
\end{align*}
Difference in expected payment in round $t$ satisfies:
\begin{align*}
\payment - \bar \payment = p_1 \cdot \conservBid_1 \cdot [\CDF(\conservBid_1) - \optimisticCDF(\conservBid_1)] + p_2 \cdot \conservBid_2 \cdot [\CDF(\conservBid_2) - \optimisticCDF(\conservBid_2)] \le 0 ~.
\end{align*}
Therefore, over $T$ rounds, the expected reward satisfies:
\begin{align*}
    \TotalReward(\ConsAlg(\conservCDF) \mid (\CDF, \valSeq^T))  & = \sum\nolimits_{t \in [T]}  \reward_t \\
     & \ge \sum\nolimits_{t \in [T]} \bar \reward_t -T \cdot 2\eps \\ 
     & = \TotalReward(\ConsAlg(\conservCDF) \mid (\optimisticCDF, \valSeq^T)) - T \cdot 2 \epsilon ~.
\end{align*}
Similarly, the ROI violation over $T$ rounds satisfies:
\begin{align*}
    \TotalROIViolation(\ConsAlg(\conservCDF) \mid (\CDF, \valSeq^T))  & = \TotalPayment(\ConsAlg(\conservCDF) \mid (\CDF, \valSeq^T)) - \TotalReward(\ConsAlg(\conservCDF) \mid (\CDF, \valSeq^T)) \\
     & \le \TotalPayment(\ConsAlg(\conservCDF) \mid (\optimisticCDF, \valSeq^T)) - \TotalReward(\ConsAlg(\conservCDF) \mid (\CDF, \valSeq^T)) + T \cdot 2 \eps  \\
     & = \Delta(\ConsAlg(\conservCDF) \mid (\optimisticCDF, \valSeq^T)) + T \cdot 2 \eps ~. \qedhere
\end{align*}
\end{proof}


We analyze the Regret and ROI violations of the conservative strategy based on \Cref{alg:randomized_strategy_under_F} with an estimated distribution against the Optimal strategy under true distribution.
\begin{lemma} \label{lem:bandit_feedback_estimated_distribution_aganst_optimal_strategy}
    Consider a fixed value sequence $\valSeq^T$ and a true CDF $\CDF$ of the maximum competing bid. Suppose that after the exploration phase of \Cref{alg:randomized_strategy_with_bandit_feedback}, our empirical estimate $\empiricalCDF$ satisfies $|\CDF(c_k) - \empiricalCDF(c_k)| \le \eps$, for all grid points $ \{c_k\}_{k=0}^K(c_k=\sfrac{k}{K})$. We define the optimistic distribution $\optimisticCDF(\bid) = \min(\empiricalCDF(\bid) + \eps, 1)$ and its conservative version $\conservCDF(c_k) = \optimisticCDF(\min(c_k + \sfrac{1}{K}, 1))$.
    Then, run the conservative strategy proposed in \Cref{sec:conservative_strategy} with $\conservCDF$ as input, denoted as $\ConsAlg(\conservCDF)$, for $T$ rounds. The total reward satisfies:
    \[
        \TotalReward(\ConsAlg(\conservCDF) \mid (\CDF, \valSeq^T)) \ge \TotalReward(\opt(\CDF) \mid (\CDF, \valSeq^T)) - O(\sqrt{T}) - T \cdot 2 \eps ~.\]
    The ROI violation is bounded by:
    \[
        \Delta(\ConsAlg(\conservCDF) \mid (\CDF, \valSeq^T)) \le 2\sqrt{T}\log T + T\cdot 2\eps + \frac{T}{K} ~.
    \]
\end{lemma}
\begin{proof}[Proof of \Cref{lem:bandit_feedback_estimated_distribution_aganst_optimal_strategy}]
Given $|\CDF(c_k) - \empiricalCDF(c_k)| \le \eps$, and $\optimisticCDF(c_k) = \min(\empiricalCDF(c_k) + \eps, 1)$, we have $\CDF(c_k) \le \optimisticCDF(c_k)$. Also, $|\optimisticCDF(c_k) - \CDF(c_k)| \le 2\eps$ for all grid points $\{c_k\}_{k=0}^K~$.

From \Cref{lem:bandit_feedback_algorithm_with_true_distribution_vs_estimated_distribution}, we have:
\[
   \TotalReward(\ConsAlg(\conservCDF) \mid (\CDF, \valSeq^T)) \ge \TotalReward(\ConsAlg(\conservCDF) \mid (\optimisticCDF, \valSeq^T)) - T\cdot 2\eps ~. 
\]

From \Cref{lem:conservative_strategy_vs_randomized_algorithm}, we know:
\[
    \TotalReward(\ConsAlg(\conservCDF) \mid (\optimisticCDF, \valSeq^T)) = \TotalReward(\Algorithm(\conservCDF) \mid (\conservCDF, \valSeq^T)) ~.
    \]

\Cref{prop:performance algo known F} states that the Regret of \Cref{alg:randomized_strategy_under_F} is $O(\sqrt{T})$:
\[
    \TotalReward(\Algorithm(\conservCDF) \mid (\conservCDF, \valSeq^T))=\TotalReward(\opt(\conservCDF) \mid (\conservCDF, \valSeq^T)) - O(\sqrt T) ~.
\]

Since \Cref{lem:bandit_feedback_F_vs_bar_F_prime} showed that $\conservCDF(\bid) \ge \CDF(\bid)$ over the entire range $\bid \in [0,1]$ and \Cref{lem:optimal_strategy_under_two_different_Fs} implies that the optimal reward of $\conservCDF$ will not be worse than $\CDF$, 
\[
    \TotalReward( \opt(\conservCDF) \mid ( \conservCDF , \valSeq^T)) \ge \TotalReward (\opt(\CDF) \mid (\CDF, \valSeq^T))~. 
\]
Combining these, we have:
\begin{align*}
   \TotalReward(\ConsAlg(\conservCDF) \mid (\CDF , \valSeq^T))  & \ge \TotalReward(\ConsAlg(\conservCDF) \mid (\optimisticCDF , \valSeq^T)) - T\cdot 2\eps \\
    & = \TotalReward(\Algorithm(\conservCDF) \mid (\conservCDF, \valSeq^T)) - T\cdot 2\eps \\
    & = \TotalReward(\opt(\conservCDF) \mid (\conservCDF , \valSeq^T)) - T\cdot 2\eps - O(\sqrt{T}) \\ 
    & \ge \TotalReward(\opt(\CDF) \mid (\CDF , \valSeq^T)) - T\cdot 2\eps - O(\sqrt{T}) ~.
\end{align*}
Similarly, for the ROI violation, from \Cref{lem:bandit_feedback_algorithm_with_true_distribution_vs_estimated_distribution}, we have:
\[
    \Delta(\ConsAlg(\conservCDF) \mid ( \CDF, \valSeq^T)) \le \Delta(\ConsAlg(\conservCDF) \mid (\optimisticCDF, \valSeq^T)) + T \cdot 2 \eps~.
\]
From \Cref{lem:conservative_strategy_vs_randomized_algorithm}, we have:
\[
    \TotalROIViolation(\ConsAlg(\conservCDF) \mid (\optimisticCDF, \valSeq^T)) \le \TotalROIViolation(\Algorithm(\conservCDF) \mid (\conservCDF, \valSeq^T)) + \frac{T}{K} ~.
\]
From \Cref{prop:performance algo known F}, we have:
\[   
   \Delta(\Algorithm(\optimisticCDF) \mid(\optimisticCDF, \valSeq^T)) \le 2\sqrt{T}\log T ~.
\]
Combining above, the ROI violation is bounded by: 
\begin{align*}
    \Delta(\ConsAlg(\conservCDF) \mid (\CDF, \valSeq^T))  & \le \Delta(\ConsAlg(\conservCDF) \mid (\optimisticCDF, \valSeq^T)) + T \cdot 2 \eps \\
     & \le \TotalROIViolation(\Algorithm(\conservCDF) \mid (\conservCDF, \valSeq^T)) + T \cdot 2 \eps + \frac{T}{K} \\
     & \le 2\sqrt{T}\log T + T \cdot 2 \eps + \frac{T}{K} ~. \qedhere
\end{align*}
\end{proof}
\subsection{Regrets and ROI Violation Bound of the \Cref{alg:randomized_strategy_with_bandit_feedback}} \label{sec:bandit_feedback-regret_ros_analysis}
This section derives the total regret and ROI violation of the proposed \Cref{alg:randomized_strategy_with_bandit_feedback} by considering both the exploration and exploitation phases.
\begin{lemma}\label{lem:bandit_feedback-regret_ros}
Consider an unknown CDF $\CDF$ of maximum competing bid. Given a fixed value sequence $\valSeq^T$, run \Cref{alg:randomized_strategy_with_bandit_feedback} for $T$ rounds. We sample $M$ times for each grid point in the point set $\{c_k\}_{k=0}^{K-1}$ in the exploration phase. The reward of \Cref{alg:randomized_strategy_with_bandit_feedback} satisfies 
\begin{equation}
\begin{aligned}
  \TotalRegret(\mathrm{\Cref{alg:randomized_strategy_with_bandit_feedback}} \mid (\CDF, \valSeq^T)) \le M \cdot K + O(\sqrt{T})  + T\cdot 2 \eps ~.
\end{aligned}
\end{equation}
The ROI violation satisfies
\begin{equation}
\begin{aligned}
    \TotalROIViolation(\mathrm{\Cref{alg:randomized_strategy_with_bandit_feedback}} \mid (\CDF, \valSeq^T)) \le M\cdot K + 2\sqrt{T}\log T + T \cdot 2 \eps + \frac{T}{K} ~,
\end{aligned}
\end{equation}
where $\eps$ is a parameter such that the empirical CDF $\empiricalCDF$ obtained in the exploration phase satisfies: $\prob{\exists i \in [K]: |\empiricalCDF(c_k) -  \CDF(c_k)| \ge \eps} \le K \cdot \exp(-2M\eps^2)~$.
\end{lemma}
\begin{proof}[Proof of \Cref{lem:bandit_feedback-regret_ros}]
The algorithm operates in two phases. In the exploration phase, lasting $T_1=M\cdot K$ rounds, we bid at $K$ different prices. In the worst case, the regret in each of these rounds could be as high as the potential reward, which is bounded by 1. Thus, the regret of the exploration phase, denoted by $\TotalRegret_1$, is bounded by 
\[
    \TotalRegret_1 \le \TotalReward(\opt(\CDF) \mid (\CDF, \valSeq^{T_1})) \le \sum\nolimits_{ i \in [M \cdot K]} v_i \le M \cdot K~,
\] 
where $V^{T_1}$ is the value sequence in the exploration phase.
Similarly, the total payment, denoted by $\TotalPayment_1$, is bounded by
\[
    \TotalPayment_1 =\sum\nolimits_{t \in [T_1]}b_t\cdot F(b_t) \le T_1 \le M\cdot K~.
\]
In the worst case for the ROI violation, it will be as high as payment:
\[\TotalROIViolation_1 \le \TotalPayment_1 - \TotalReward_1 \le T_1 \le M\cdot K~.\]
In the exploitation phase, lasting $T_2=1-M \cdot K$ rounds, we employ the conservative strategy $\ConsAlg$ based on the conservative version of the optimistic CDF $\conservCDF$. From \Cref{lem:bandit_feedback_estimated_distribution_aganst_optimal_strategy} we obtain that the reward of the conservative strategy satisfies: 
\begin{align*}
  \TotalReward(\ConsAlg(\optimisticCDF) \mid (\CDF, \valSeq^{T_2}))  & \ge \TotalReward(\opt(\CDF) \mid (\CDF, \valSeq^{T_2})) -  o(\sqrt{T_2}) - T_2 \cdot 2 \eps  \\ 
   & \ge \TotalReward(\opt(\CDF)\mid (\CDF, \valSeq^{T_2})) - O(\sqrt{T}) - T \cdot 2 \eps ~,
\end{align*}
where $V^{T_2}$ is the value sequence in the exploitation phase.

The regret of the exploitation phase, denoted by $\TotalRegret_2$, is bounded by
\[\TotalRegret_2 = \TotalReward(\opt(\CDF)\mid (\CDF, \valSeq^{T_2})) - \TotalReward(\ConsAlg(\optimisticCDF)\mid (\CDF, \valSeq^{T_2}))  \le O(\sqrt{T}) + T \cdot 2 \eps ~.\]
The ROI violoation of the exploitation phase, denoted by $\TotalROIViolation_2$, satisfies
\[\TotalROIViolation_2\le 2\sqrt{T_2}\log T_2 + T_2 \cdot 2 \eps + \frac{T_2}{K} ~.\]
Combining the regret from both phases, we have:
\[
    \TotalRegret(\mathrm{\Cref{alg:randomized_strategy_with_bandit_feedback}} \mid (\CDF, \valSeq^T)) \le M \cdot K + O(\sqrt{T}) + T \cdot 2 \eps ~.
\]
Combining the ROI violation from both phases, we have:
\[ \TotalROIViolation(\mathrm{\Cref{alg:randomized_strategy_with_bandit_feedback}} \mid (\CDF, \valSeq^T)) \le M\cdot K + 2\sqrt{T}\log T + T \cdot 2 \eps + \frac{T}{K} ~. \qedhere\]
\end{proof}

\subsection{Put All Pieces Together to Prove \Cref{thm:bandit_feedback-regret_ros}}

\begin{proof}[Proof of \Cref{thm:bandit_feedback-regret_ros}] \label{proof:bandit_feedback-regret_ros}
    We first discuss the regret and ROI violation incurred with any fixed value sequence $\valSeq^T$.
    The proposed algorithm splits into 2 phases: Exploration and Exploitation. In the exploration phase, we identify $K+1$ grid points $c_k=\sfrac{k}{K}$ for $k=0,\ldots,K$. We then sample the points $\{c_k\}_{k=0}^{K-1}$ $M$ times each, which leads to $M\cdot K$ rounds. The point $c_K=1$ is generally not samples as its outcome $\CDF(1)=1$ is assumed known. 
    
    By \Cref{lem:bandit_feedback-regret_ros}, The total Regret combining two phases is bounded by
\[\TotalRegret(\mathrm{\Cref{alg:randomized_strategy_with_bandit_feedback}} \mid (\CDF, \valSeq^T)) \le M \cdot K + O(\sqrt{T}) + T \cdot 2 \eps~,\]
with ROI violation:
\[ \TotalROIViolation(\mathrm{\Cref{alg:randomized_strategy_with_bandit_feedback}} \mid (\CDF, \valSeq^T)) \le M\cdot K + 2\sqrt{T}\log T + T \cdot 2 \eps + \frac{T}{K}~,\]
 where $\eps$ is the parameter to measure the accuracy of the empirical CDF $\empiricalCDF$ we constructed in \Cref{sec:bandit_feedback_contruct_empirical_distribution} relative to the true CDF $\CDF$. From \Cref{lem:bandit_feedback-CDF_bound}, they are bounded by
\[
        \prob{\exists k \in [K]: |\empiricalCDF(c_k) - \CDF(c_k)| \ge \varepsilon} \le K \cdot \exp(-2M\eps^2)~.
 \]
For each grid point, it satisfies:
\[
            \prob{|\empiricalCDF(c_k) - \CDF(c_k)| \ge \varepsilon} \le \exp(-2 M \varepsilon^2) ~,
    \]
Denote $\exp(-2 M \varepsilon^2)$ by $\delta$. We need to choose $\eps$ such that the probability of large deviation in the empirical CDF is small. Let's aim for a failure probability $\delta = T^{-c},~c > 0 $.
To guarantee that the union bound will decrease, 
     \[K = o\left(\frac{1}{\delta}\right)= o(T^\alpha)~. \]
From \Cref{lem:bandit_feedback-regret_ros}
     \begin{align*}
         \TotalROIViolation(\mathrm{\Cref{alg:randomized_strategy_with_bandit_feedback}} \mid (\CDF, \valSeq^T))  & \le M\cdot K + 2\sqrt{T}\log T + T \cdot 2 \eps + \frac{T}{K} \\
          & \le M\cdot T^\alpha + 2\sqrt{T}\log T + T \cdot 2 \sqrt{\alpha \cdot \frac{\log T}{2M}} + T^{1-\alpha} ~. 
     \end{align*} 
     Therefore, The ROI violation will be bounded by one of these: $M\cdot T^\alpha$, $T \cdot 2 \sqrt{\alpha \cdot \frac{\log T}{2M}}$, and $T^{1-\alpha}$. The minimum bound will be obtained when
     \[
         O(M\cdot T^\alpha)=O\left(\frac{T}{\sqrt{M}}\right)=O(T^{1-\alpha})~.
     \]
     Thus, with choosing $M=T^{\sfrac{1}{2}}, K = T^{\sfrac{1}{4}}$, the ROI violation can be bounded as:
     \begin{align*}
         \TotalROIViolation( \mathrm{\Cref{alg:randomized_strategy_with_bandit_feedback}} \mid (\CDF, \valSeq^T))  & \le T^{\sfrac{3}{4}} + 2\sqrt{T}\log T + T^{\sfrac{3}{4}} \cdot 2 \sqrt{\alpha \cdot \log T} + T^{\sfrac{3}{4}} \\
          & = \widetilde O (T^{\sfrac{3}{4}})~. 
     \end{align*}
     Substitute $M=T^{\sfrac{1}{2}}, K = T^{\sfrac{1}{4}}$ into the regret bound:
     \begin{align*}
         \TotalRegret( \mathrm{\Cref{alg:randomized_strategy_with_bandit_feedback}} \mid (\CDF, \valSeq^T))  & \le M \cdot K + O(\sqrt{T}) + T \cdot 2 \eps \\
          & \le T^{\sfrac{3}{4}} + O(\sqrt{T}) +  T^{\sfrac{3}{4}} \cdot 2 \sqrt{\alpha \cdot \log T} \\
          & = \widetilde O (T^{\sfrac{3}{4}})~.
     \end{align*}
Thus, when a value sequence $\valSeq^T$ is i.i.d drawn from $\valDist^T$,  the expected Regret and ROI violation over the distribution $\CDF$ and $\valDist$ satisfies
\begin{align*}
    \TotalRegret(\mathrm{\Cref{alg:randomized_strategy_with_bandit_feedback}} \mid (\competingbidDist, \valDist))
    & = \expect[\valSeq^T \sim \valDist^T]{\TotalRegret(\mathrm{\Cref{alg:randomized_strategy_with_bandit_feedback}} \mid (\competingbidDist, \valSeq^T))} 
    = \widetilde O (T^{\sfrac{3}{4}})~,\\
    \TotalROIViolation(\mathrm{\Cref{alg:randomized_strategy_with_bandit_feedback}} \mid (\competingbidDist, \valDist))
    & = \expect[\valSeq^T \sim \valDist^T]{\TotalROIViolation(\mathrm{\Cref{alg:randomized_strategy_with_bandit_feedback}} \mid (\competingbidDist, \valSeq^T))}
    = \widetilde O (T^{\sfrac{3}{4}})~. \qedhere
\end{align*}
\end{proof}

%% file: 0-main.bbl

\begin{thebibliography}{35}


\ifx \showCODEN    \undefined \def \showCODEN     #1{\unskip}     \fi
\ifx \showISBNx    \undefined \def \showISBNx     #1{\unskip}     \fi
\ifx \showISBNxiii \undefined \def \showISBNxiii  #1{\unskip}     \fi
\ifx \showISSN     \undefined \def \showISSN      #1{\unskip}     \fi
\ifx \showLCCN     \undefined \def \showLCCN      #1{\unskip}     \fi
\ifx \shownote     \undefined \def \shownote      #1{#1}          \fi
\ifx \showarticletitle \undefined \def \showarticletitle #1{#1}   \fi
\ifx \showURL      \undefined \def \showURL       {\relax}        \fi
\providecommand\bibfield[2]{#2}
\providecommand\bibinfo[2]{#2}
\providecommand\natexlab[1]{#1}
\providecommand\showeprint[2][]{arXiv:#2}

\bibitem[Aggarwal et~al\mbox{.}(2024)]%
        {aggarwal2024auto}
\bibfield{author}{\bibinfo{person}{Gagan Aggarwal}, \bibinfo{person}{Ashwinkumar Badanidiyuru}, \bibinfo{person}{Santiago~R Balseiro}, \bibinfo{person}{Kshipra Bhawalkar}, \bibinfo{person}{Yuan Deng}, \bibinfo{person}{Zhe Feng}, \bibinfo{person}{Gagan Goel}, \bibinfo{person}{Christopher Liaw}, \bibinfo{person}{Haihao Lu}, \bibinfo{person}{Mohammad Mahdian}, {et~al\mbox{.}}} \bibinfo{year}{2024}\natexlab{}.
\newblock \showarticletitle{Auto-bidding and auctions in online advertising: A survey}.
\newblock \bibinfo{journal}{\emph{ACM SIGecom Exchanges}} \bibinfo{volume}{22}, \bibinfo{number}{1} (\bibinfo{year}{2024}), \bibinfo{pages}{159--183}.
\newblock


\bibitem[Aggarwal et~al\mbox{.}(2019)]%
        {aggarwal2019autobidding}
\bibfield{author}{\bibinfo{person}{Gagan Aggarwal}, \bibinfo{person}{Ashwinkumar Badanidiyuru}, {and} \bibinfo{person}{Aranyak Mehta}.} \bibinfo{year}{2019}\natexlab{}.
\newblock \showarticletitle{Autobidding with constraints}. In \bibinfo{booktitle}{\emph{Web and Internet Economics: 15th International Conference, WINE 2019, New York, NY, USA, December 10--12, 2019, Proceedings 15}}. Springer, \bibinfo{pages}{17--30}.
\newblock


\bibitem[Aggarwal et~al\mbox{.}(2025)]%
        {aggarwal2025no}
\bibfield{author}{\bibinfo{person}{Gagan Aggarwal}, \bibinfo{person}{Giannis Fikioris}, {and} \bibinfo{person}{Mingfei Zhao}.} \bibinfo{year}{2025}\natexlab{}.
\newblock \showarticletitle{No-regret algorithms in non-truthful auctions with budget and roi constraints}. In \bibinfo{booktitle}{\emph{Proceedings of the ACM on Web Conference 2025}}. \bibinfo{pages}{1398--1415}.
\newblock


\bibitem[Balseiro et~al\mbox{.}(2024b)]%
        {balseiro2024optimal}
\bibfield{author}{\bibinfo{person}{Santiago Balseiro}, \bibinfo{person}{Yuan Deng}, \bibinfo{person}{Jieming Mao}, \bibinfo{person}{Vahab Mirrokni}, {and} \bibinfo{person}{Song Zuo}.} \bibinfo{year}{2024}\natexlab{b}.
\newblock \showarticletitle{Optimal Mechanisms for a Value Maximizer: The Futility of Screening Targets}. In \bibinfo{booktitle}{\emph{Proceedings of the 25th ACM Conference on Economics and Computation}}. \bibinfo{pages}{1101--1101}.
\newblock


\bibitem[Balseiro et~al\mbox{.}(2019)]%
        {balseiro2019contextual}
\bibfield{author}{\bibinfo{person}{Santiago Balseiro}, \bibinfo{person}{Negin Golrezaei}, \bibinfo{person}{Mohammad Mahdian}, \bibinfo{person}{Vahab Mirrokni}, {and} \bibinfo{person}{Jon Schneider}.} \bibinfo{year}{2019}\natexlab{}.
\newblock \showarticletitle{Contextual bandits with cross-learning}.
\newblock \bibinfo{journal}{\emph{Advances in Neural Information Processing Systems}}  \bibinfo{volume}{32} (\bibinfo{year}{2019}).
\newblock


\bibitem[Balseiro et~al\mbox{.}(2020)]%
        {balseiro2020dual}
\bibfield{author}{\bibinfo{person}{Santiago Balseiro}, \bibinfo{person}{Haihao Lu}, {and} \bibinfo{person}{Vahab Mirrokni}.} \bibinfo{year}{2020}\natexlab{}.
\newblock \showarticletitle{Dual mirror descent for online allocation problems}. In \bibinfo{booktitle}{\emph{International Conference on Machine Learning}}. PMLR, \bibinfo{pages}{613--628}.
\newblock


\bibitem[Balseiro et~al\mbox{.}(2024a)]%
        {balseiro2024field}
\bibfield{author}{\bibinfo{person}{Santiago~R Balseiro}, \bibinfo{person}{Kshipra Bhawalkar}, \bibinfo{person}{Zhe Feng}, \bibinfo{person}{Haihao Lu}, \bibinfo{person}{Vahab Mirrokni}, \bibinfo{person}{Balasubramanian Sivan}, {and} \bibinfo{person}{Di Wang}.} \bibinfo{year}{2024}\natexlab{a}.
\newblock \showarticletitle{A Field Guide for Pacing Budget and ROS Constraints}. In \bibinfo{booktitle}{\emph{International Conference on Machine Learning}}. PMLR, \bibinfo{pages}{2607--2638}.
\newblock


\bibitem[Balseiro et~al\mbox{.}(2021)]%
        {balseiro2021landscape}
\bibfield{author}{\bibinfo{person}{Santiago~R Balseiro}, \bibinfo{person}{Yuan Deng}, \bibinfo{person}{Jieming Mao}, \bibinfo{person}{Vahab~S Mirrokni}, {and} \bibinfo{person}{Song Zuo}.} \bibinfo{year}{2021}\natexlab{}.
\newblock \showarticletitle{The landscape of auto-bidding auctions: Value versus utility maximization}. In \bibinfo{booktitle}{\emph{Proceedings of the 22nd ACM Conference on Economics and Computation}}. \bibinfo{pages}{132--133}.
\newblock


\bibitem[Balseiro and Gur(2019)]%
        {balseiro2019learning}
\bibfield{author}{\bibinfo{person}{Santiago~R Balseiro} {and} \bibinfo{person}{Yonatan Gur}.} \bibinfo{year}{2019}\natexlab{}.
\newblock \showarticletitle{Learning in repeated auctions with budgets: Regret minimization and equilibrium}.
\newblock \bibinfo{journal}{\emph{Management Science}} \bibinfo{volume}{65}, \bibinfo{number}{9} (\bibinfo{year}{2019}), \bibinfo{pages}{3952--3968}.
\newblock


\bibitem[Bei et~al\mbox{.}(2025)]%
        {bei2025optimal}
\bibfield{author}{\bibinfo{person}{Xiaohui Bei}, \bibinfo{person}{Pinyan Lu}, \bibinfo{person}{Zhiqi Wang}, \bibinfo{person}{Tao Xiao}, {and} \bibinfo{person}{Xiang Yan}.} \bibinfo{year}{2025}\natexlab{}.
\newblock \showarticletitle{Optimal Auction Design for Mixed Bidders}. In \bibinfo{booktitle}{\emph{Proceedings of the AAAI Conference on Artificial Intelligence}}, Vol.~\bibinfo{volume}{39}. \bibinfo{pages}{13614--13621}.
\newblock


\bibitem[Borgs et~al\mbox{.}(2007)]%
        {borgs2007dynamics}
\bibfield{author}{\bibinfo{person}{Christian Borgs}, \bibinfo{person}{Jennifer Chayes}, \bibinfo{person}{Nicole Immorlica}, \bibinfo{person}{Kamal Jain}, \bibinfo{person}{Omid Etesami}, {and} \bibinfo{person}{Mohammad Mahdian}.} \bibinfo{year}{2007}\natexlab{}.
\newblock \showarticletitle{Dynamics of bid optimization in online advertisement auctions}. In \bibinfo{booktitle}{\emph{Proceedings of the 16th international conference on World Wide Web}}. \bibinfo{pages}{531--540}.
\newblock


\bibitem[Castiglioni et~al\mbox{.}(2022)]%
        {castiglioni2022online}
\bibfield{author}{\bibinfo{person}{Matteo Castiglioni}, \bibinfo{person}{Andrea Celli}, {and} \bibinfo{person}{Christian Kroer}.} \bibinfo{year}{2022}\natexlab{}.
\newblock \showarticletitle{Online learning with knapsacks: the best of both worlds}. In \bibinfo{booktitle}{\emph{International Conference on Machine Learning}}. PMLR, \bibinfo{pages}{2767--2783}.
\newblock


\bibitem[Castiglioni et~al\mbox{.}(2024)]%
        {castiglioni2024online}
\bibfield{author}{\bibinfo{person}{Matteo Castiglioni}, \bibinfo{person}{Andrea Celli}, {and} \bibinfo{person}{Christian Kroer}.} \bibinfo{year}{2024}\natexlab{}.
\newblock \showarticletitle{Online Learning under Budget and ROI Constraints via Weak Adaptivity}. In \bibinfo{booktitle}{\emph{International Conference on Machine Learning}}. PMLR, \bibinfo{pages}{5792--5816}.
\newblock


\bibitem[Cesa-Bianchi et~al\mbox{.}(2024)]%
        {cesa2024role}
\bibfield{author}{\bibinfo{person}{Nicol{\`o} Cesa-Bianchi}, \bibinfo{person}{Tommaso Cesari}, \bibinfo{person}{Roberto Colomboni}, \bibinfo{person}{Federico Fusco}, {and} \bibinfo{person}{Stefano Leonardi}.} \bibinfo{year}{2024}\natexlab{}.
\newblock \showarticletitle{The role of transparency in repeated first-price auctions with unknown valuations}. In \bibinfo{booktitle}{\emph{Proceedings of the 56th Annual ACM Symposium on Theory of Computing}}. \bibinfo{pages}{225--236}.
\newblock


\bibitem[Deng et~al\mbox{.}(2024a)]%
        {deng2024gsp}
\bibfield{author}{\bibinfo{person}{Yuan Deng}, \bibinfo{person}{Mohammad Mahdian}, \bibinfo{person}{Jieming Mao}, \bibinfo{person}{Vahab Mirrokni}, \bibinfo{person}{Hanrui Zhang}, {and} \bibinfo{person}{Song Zuo}.} \bibinfo{year}{2024}\natexlab{a}.
\newblock \showarticletitle{Efficiency of the generalized second-price auction for value maximizers}. In \bibinfo{booktitle}{\emph{Proceedings of the ACM Web Conference 2024}}. \bibinfo{pages}{46--56}.
\newblock


\bibitem[Deng et~al\mbox{.}(2024b)]%
        {deng2024efficiency}
\bibfield{author}{\bibinfo{person}{Yuan Deng}, \bibinfo{person}{Jieming Mao}, \bibinfo{person}{Vahab Mirrokni}, \bibinfo{person}{Hanrui Zhang}, {and} \bibinfo{person}{Song Zuo}.} \bibinfo{year}{2024}\natexlab{b}.
\newblock \showarticletitle{Efficiency of the first-price auction in the autobidding world}.
\newblock \bibinfo{journal}{\emph{Advances in Neural Information Processing Systems}}  \bibinfo{volume}{37} (\bibinfo{year}{2024}), \bibinfo{pages}{139270--139293}.
\newblock


\bibitem[Deng et~al\mbox{.}(2022)]%
        {deng2022posted}
\bibfield{author}{\bibinfo{person}{Yuan Deng}, \bibinfo{person}{Vahab Mirrokni}, {and} \bibinfo{person}{Hanrui Zhang}.} \bibinfo{year}{2022}\natexlab{}.
\newblock \showarticletitle{Posted pricing and dynamic prior-independent mechanisms with value maximizers}.
\newblock \bibinfo{journal}{\emph{Advances in Neural Information Processing Systems}}  \bibinfo{volume}{35} (\bibinfo{year}{2022}), \bibinfo{pages}{24158--24169}.
\newblock


\bibitem[Feng et~al\mbox{.}(2024)]%
        {feng2024strategic}
\bibfield{author}{\bibinfo{person}{Yiding Feng}, \bibinfo{person}{Brendan Lucier}, {and} \bibinfo{person}{Aleksandrs Slivkins}.} \bibinfo{year}{2024}\natexlab{}.
\newblock \showarticletitle{Strategic budget selection in a competitive autobidding world}. In \bibinfo{booktitle}{\emph{Proceedings of the 56th Annual ACM Symposium on Theory of Computing}}. \bibinfo{pages}{213--224}.
\newblock


\bibitem[Feng et~al\mbox{.}(2023)]%
        {feng2023online}
\bibfield{author}{\bibinfo{person}{Zhe Feng}, \bibinfo{person}{Swati Padmanabhan}, {and} \bibinfo{person}{Di Wang}.} \bibinfo{year}{2023}\natexlab{}.
\newblock \showarticletitle{Online Bidding Algorithms for Return-on-Spend Constrained Advertisers}. In \bibinfo{booktitle}{\emph{Proceedings of the ACM Web Conference 2023}}. \bibinfo{pages}{3550--3560}.
\newblock


\bibitem[Fikioris and Tardos(2023a)]%
        {fikioris2023approximately}
\bibfield{author}{\bibinfo{person}{Giannis Fikioris} {and} \bibinfo{person}{{\'E}va Tardos}.} \bibinfo{year}{2023}\natexlab{a}.
\newblock \showarticletitle{Approximately stationary bandits with knapsacks}. In \bibinfo{booktitle}{\emph{The Thirty Sixth Annual Conference on Learning Theory}}. PMLR, \bibinfo{pages}{3758--3782}.
\newblock


\bibitem[Fikioris and Tardos(2023b)]%
        {fikioris2023liquid}
\bibfield{author}{\bibinfo{person}{Giannis Fikioris} {and} \bibinfo{person}{{\'E}va Tardos}.} \bibinfo{year}{2023}\natexlab{b}.
\newblock \showarticletitle{Liquid welfare guarantees for no-regret learning in sequential budgeted auctions}. In \bibinfo{booktitle}{\emph{Proceedings of the 24th ACM Conference on Economics and Computation}}. \bibinfo{pages}{678--698}.
\newblock


\bibitem[Gaitonde et~al\mbox{.}(2023)]%
        {gaitonde2023budget}
\bibfield{author}{\bibinfo{person}{Jason Gaitonde}, \bibinfo{person}{Yingkai Li}, \bibinfo{person}{Bar Light}, \bibinfo{person}{Brendan Lucier}, {and} \bibinfo{person}{Aleksandrs Slivkins}.} \bibinfo{year}{2023}\natexlab{}.
\newblock \showarticletitle{Budget Pacing in Repeated Auctions: Regret and Efficiency Without Convergence}. In \bibinfo{booktitle}{\emph{14th Innovations in Theoretical Computer Science Conference (ITCS 2023)}}, Vol.~\bibinfo{volume}{251}. Schloss Dagstuhl--Leibniz-Zentrum f\"ur Informatik, \bibinfo{pages}{52}.
\newblock


\bibitem[Han et~al\mbox{.}(2020)]%
        {han2020learning}
\bibfield{author}{\bibinfo{person}{Yanjun Han}, \bibinfo{person}{Zhengyuan Zhou}, \bibinfo{person}{Aaron Flores}, \bibinfo{person}{Erik Ordentlich}, {and} \bibinfo{person}{Tsachy Weissman}.} \bibinfo{year}{2020}\natexlab{}.
\newblock \showarticletitle{Learning to bid optimally and efficiently in adversarial first-price auctions}.
\newblock \bibinfo{journal}{\emph{arXiv preprint arXiv:2007.04568}} (\bibinfo{year}{2020}).
\newblock


\bibitem[Immorlica et~al\mbox{.}(2022)]%
        {immorlica2022adversarial}
\bibfield{author}{\bibinfo{person}{Nicole Immorlica}, \bibinfo{person}{Karthik Sankararaman}, \bibinfo{person}{Robert Schapire}, {and} \bibinfo{person}{Aleksandrs Slivkins}.} \bibinfo{year}{2022}\natexlab{}.
\newblock \showarticletitle{Adversarial bandits with knapsacks}.
\newblock \bibinfo{journal}{\emph{J. ACM}} \bibinfo{volume}{69}, \bibinfo{number}{6} (\bibinfo{year}{2022}), \bibinfo{pages}{1--47}.
\newblock


\bibitem[Kesselheim and Singla(2020)]%
        {kesselheim2020online}
\bibfield{author}{\bibinfo{person}{Thomas Kesselheim} {and} \bibinfo{person}{Sahil Singla}.} \bibinfo{year}{2020}\natexlab{}.
\newblock \showarticletitle{Online learning with vector costs and bandits with knapsacks}. In \bibinfo{booktitle}{\emph{Conference on Learning Theory}}. PMLR, \bibinfo{pages}{2286--2305}.
\newblock


\bibitem[Kumar et~al\mbox{.}(2024)]%
        {kumar2024strategically}
\bibfield{author}{\bibinfo{person}{Rachitesh Kumar}, \bibinfo{person}{Jon Schneider}, {and} \bibinfo{person}{Balasubramanian Sivan}.} \bibinfo{year}{2024}\natexlab{}.
\newblock \showarticletitle{Strategically-Robust Learning Algorithms for Bidding in First-Price Auctions}. In \bibinfo{booktitle}{\emph{Proceedings of the 25th ACM Conference on Economics and Computation}}. \bibinfo{pages}{893--893}.
\newblock


\bibitem[Liaw et~al\mbox{.}(2023)]%
        {liaw2023efficiency}
\bibfield{author}{\bibinfo{person}{Christopher Liaw}, \bibinfo{person}{Aranyak Mehta}, {and} \bibinfo{person}{Andres Perlroth}.} \bibinfo{year}{2023}\natexlab{}.
\newblock \showarticletitle{Efficiency of non-truthful auctions in auto-bidding: The power of randomization}. In \bibinfo{booktitle}{\emph{Proceedings of the ACM Web Conference 2023}}. \bibinfo{pages}{3561--3571}.
\newblock


\bibitem[Lu et~al\mbox{.}(2023)]%
        {lu2023auction}
\bibfield{author}{\bibinfo{person}{Pinyan Lu}, \bibinfo{person}{Chenyang Xu}, {and} \bibinfo{person}{Ruilong Zhang}.} \bibinfo{year}{2023}\natexlab{}.
\newblock \showarticletitle{Auction design for value maximizers with budget and return-on-spend constraints}. In \bibinfo{booktitle}{\emph{International Conference on Web and Internet Economics}}. Springer, \bibinfo{pages}{474--491}.
\newblock


\bibitem[Lucier et~al\mbox{.}(2024)]%
        {lucier2024autobidders}
\bibfield{author}{\bibinfo{person}{Brendan Lucier}, \bibinfo{person}{Sarath Pattathil}, \bibinfo{person}{Aleksandrs Slivkins}, {and} \bibinfo{person}{Mengxiao Zhang}.} \bibinfo{year}{2024}\natexlab{}.
\newblock \showarticletitle{Autobidders with budget and roi constraints: Efficiency, regret, and pacing dynamics}. In \bibinfo{booktitle}{\emph{The Thirty Seventh Annual Conference on Learning Theory}}. PMLR, \bibinfo{pages}{3642--3643}.
\newblock


\bibitem[Lv et~al\mbox{.}(2023)]%
        {lv2023auction}
\bibfield{author}{\bibinfo{person}{Hongtao Lv}, \bibinfo{person}{Xiaohui Bei}, \bibinfo{person}{Zhenzhe Zheng}, {and} \bibinfo{person}{Fan Wu}.} \bibinfo{year}{2023}\natexlab{}.
\newblock \showarticletitle{Auction design for bidders with ex post roi constraints}. In \bibinfo{booktitle}{\emph{International Conference on Web and Internet Economics}}. Springer, \bibinfo{pages}{492--508}.
\newblock


\bibitem[Nedelec et~al\mbox{.}(2022)]%
        {nedelec2022learning}
\bibfield{author}{\bibinfo{person}{Thomas Nedelec}, \bibinfo{person}{Cl{\'e}ment Calauz{\`e}nes}, \bibinfo{person}{Noureddine El~Karoui}, \bibinfo{person}{Vianney Perchet}, {et~al\mbox{.}}} \bibinfo{year}{2022}\natexlab{}.
\newblock \showarticletitle{Learning in repeated auctions}.
\newblock \bibinfo{journal}{\emph{Foundations and Trends{\textregistered} in Machine Learning}} \bibinfo{volume}{15}, \bibinfo{number}{3} (\bibinfo{year}{2022}), \bibinfo{pages}{176--334}.
\newblock


\bibitem[Susan et~al\mbox{.}(2023)]%
        {susan2023multi}
\bibfield{author}{\bibinfo{person}{Fransisca Susan}, \bibinfo{person}{Negin Golrezaei}, {and} \bibinfo{person}{Okke Schrijvers}.} \bibinfo{year}{2023}\natexlab{}.
\newblock \showarticletitle{Multi-platform budget management in ad markets with non-ic auctions}.
\newblock \bibinfo{journal}{\emph{arXiv preprint arXiv:2306.07352}} (\bibinfo{year}{2023}).
\newblock


\bibitem[Vijayan et~al\mbox{.}(2025)]%
        {vijayan2025online}
\bibfield{author}{\bibinfo{person}{Sushant Vijayan}, \bibinfo{person}{Zhe Feng}, \bibinfo{person}{Swati Padmanabhan}, \bibinfo{person}{Karthikeyan Shanmugam}, \bibinfo{person}{Arun Suggala}, {and} \bibinfo{person}{Di Wang}.} \bibinfo{year}{2025}\natexlab{}.
\newblock \showarticletitle{Online Bidding under RoS Constraints without Knowing the Value}. In \bibinfo{booktitle}{\emph{Proceedings of the ACM on Web Conference 2025}}. \bibinfo{pages}{3096--3107}.
\newblock


\bibitem[Wang et~al\mbox{.}(2023)]%
        {wang2023learning}
\bibfield{author}{\bibinfo{person}{Qian Wang}, \bibinfo{person}{Zongjun Yang}, \bibinfo{person}{Xiaotie Deng}, {and} \bibinfo{person}{Yuqing Kong}.} \bibinfo{year}{2023}\natexlab{}.
\newblock \showarticletitle{Learning to bid in repeated first-price auctions with budgets}. In \bibinfo{booktitle}{\emph{International Conference on Machine Learning}}. PMLR, \bibinfo{pages}{36494--36513}.
\newblock


\bibitem[Weed et~al\mbox{.}(2016)]%
        {weed2016online}
\bibfield{author}{\bibinfo{person}{Jonathan Weed}, \bibinfo{person}{Vianney Perchet}, {and} \bibinfo{person}{Philippe Rigollet}.} \bibinfo{year}{2016}\natexlab{}.
\newblock \showarticletitle{Online learning in repeated auctions}. In \bibinfo{booktitle}{\emph{Conference on Learning Theory}}. PMLR, \bibinfo{pages}{1562--1583}.
\newblock


\end{thebibliography}
